\documentclass[11pt]{article}
\usepackage[letterpaper,margin=1.0in]{geometry}

\usepackage[T1]{fontenc}
\usepackage{amsmath, amsthm, amssymb}
\usepackage{tabularx, booktabs} 

\usepackage{mathtools}
\usepackage{exscale,relsize}
\usepackage{newpxmath}
\usepackage{inconsolata, stmaryrd}
\usepackage{tikz, graphicx, subfig, float, tcolorbox, framed, venndiagram, listings, color, enumerate, enumitem, setspace, pdfpages}  
\usepackage{sidecap}

\usepackage{hyperref}
\usepackage{cleveref}
\usepackage{thm-restate}
\usepackage{float}

\usepackage[boxruled,vlined,commentsnumbered,titlenotnumbered,linesnumbered]{algorithm2e}

\hypersetup{colorlinks,breaklinks,urlcolor=[rgb]{0,0,0},linkcolor=[rgb]{0,0,0.8},citecolor=[rgb]{0,0,0.8}}

\tcbset{
  colframe=black, 
  colback=white!90!yellow, 
  boxrule=0.5pt, 
  arc=2mm, 
  outer arc=2mm, 
  boxsep=2mm, 
}

\DeclareMathOperator{\disj}{disj}
\DeclareMathOperator{\dist}{dist}

\renewcommand{\geq}{\geqslant}
\newcommand{\A}{\mathcal{A}}
\newcommand{\B}{\mathcal{B}}
\newcommand{\C}{\mathcal{C}}

\newcommand{\G}{\mathcal{G}}

\newcommand{\cL}{\mathcal{L}}

\newcommand{\W}{\mathcal{W}}
\newcommand{\X}{\mathcal{X}}
\newcommand{\Y}{\mathcal{Y}}

\newcommand{\barz}{{\bar z}}
\newcommand{\barD}{{\bar D}}
\newcommand{\barH}{{\bar H}}
\newcommand{\barF}{{\bar F}}
\newcommand{\barFd}{{\bar F_{\disj}}}
\newcommand{\barFdk}{{\bar F_{\disj,k}}}

\newcommand{\indeg}{\text{indeg}}
\newcommand{\outdeg}{\text{outdeg}}
\newcommand{\IN}{\Gamma^{\textsc{in}}}
\newcommand{\OUT}{\Gamma^{\textsc{out}}}
\newcommand{\inedges}{{\cal E}^{\textsc{in}}}
\newcommand{\outedges}{{\cal E}^{\textsc{out}}}
\DeclareMathOperator{\depth}{\textsc{depth}}
\DeclareMathOperator{\domtree}{\textsc{DomTree}}

\SetNlSty{bfseries}{\color{black}}{}
\newtheorem{theorem}{Theorem}[section] 
\newtheorem{definition}[theorem]{Definition}

\newtheorem{lemma}[theorem]{Lemma}

\newtheorem{claim}[theorem]{Claim}

\title{
Efficient Algorithms for Disjoint Shortest Paths Problem and its Extensions
}

\author{
  Keerti Choudhary\thanks{Department of Computer Science and Engineering, IIT Delhi, India. Email: keerti@cse.iitd.ac.in} \and
  Amit Kumar\thanks{Department of Computer Science and Engineering, IIT Delhi, India. Email: amitk@cse.iitd.ac.in} \and
  Lakshay Saggi\thanks{Department of Computer Science and Engineering, IIT Delhi, India. Email: csz228231@cse.iitd.ac.in}
}

\date{}

\begin{document}
\maketitle
\thispagestyle{empty}

\begin{abstract}
We study the \emph{2-Disjoint Shortest Paths (2-DSP)} problem: given a directed weighted graph and two terminal pairs $(s_1,t_1)$ and $(s_2,t_2)$, decide whether there exist vertex-disjoint shortest paths between each pair.

Building on recent advances in disjoint shortest paths for DAGs and undirected graphs~(Akmal et al. 2024), we present an $O(mn \log n)$-time algorithm for this problem in weighted directed graphs that do not contain negative or zero weight cycles. This algorithm presents a significant improvement over the previously known $O(m^5n)$-time bound~(Berczi et al. 2017). Our approach exploits the algebraic structure of polynomials that enumerate shortest paths between terminal pairs. A key insight is that these polynomials admit a recursive decomposition, enabling efficient evaluation via dynamic programming over fields of characteristic~two. Furthermore, we demonstrate how to report the corresponding paths in $O(mn^2 \log n)$-time.

In addition, we extend our techniques to a more general setting: given two terminal pairs $(s_1, t_1)$ and $(s_2, t_2)$ in a directed graph, find the minimum possible number of vertex intersections between any shortest path from $s_1$ to $t_1$ and $s_2$  to $t_2$. We call this the \emph{Minimum 2-Disjoint Shortest Paths (Min-2-DSP)} problem. We provide in this paper the first efficient algorithm for this problem, including an $O(m^2 n^3)$-time algorithm for directed graphs with positive edge weights, and an $O(m+n)$-time algorithm for DAGs
and undirected graphs. 
Moreover, if the number of intersecting vertices is at least one, we show that it is possible to report the paths in the same $O(m+n)$-time. 
This is somewhat surprising, as there is no known $o(mn)$ time algorithm for explicitly reporting the paths if they are vertex-disjoint, and is left as an open problem in~(Akmal et al. 2024).
\end{abstract}

\newpage

\setcounter{page}{1} 
\section{Introduction} 
The {\em 2-Disjoint Shortest Paths (2-DSP)} problem seeks to determine whether a graph contains two vertex-disjoint paths such that each path is a shortest path between its designated source and target vertices. More formally, given a weighted graph $G = (V, E)$ with edge weights $\ell: E \to \mathbb{R}$ and two source-target pairs $(s_1, t_1)$ and $(s_2, t_2)$, the goal is to determine whether there exist vertex-disjoint paths $P_1$ and $P_2$, where $P_1$ is a shortest $s_1$--$t_1$ path and $P_2$ is a shortest $s_2$--$t_2$ path. 

The 2-DSP problem is a fundamental problem in combinatorial optimization with applications in network flows, routing, and circuit design~\cite{Thompson80, SrinivasM2005,  Ogier93}. Beyond its practical relevance, the problem has attracted significant attention because its complexity varies  across graph classes, giving rise to a rich algorithmic landscape. 

The 2-DSP problem has a rich history of algorithmic results. Eilam-Tzoreff~\cite{Tzoreff98} gave the first polynomial-time algorithm with running time $O(n^8)$. This was improved by Akhmedov~\cite{Akhmedov20} to $O(n^7)$ for the weighted case and $O(n^6)$ for the unweighted case. Bentert et al.~\cite{BentertNRZ21} obtained an $O(mn)$ algorithm for unweighted graphs. More recently, Akmal, Vassilevska Williams, and Wein~\cite{AkmalVW24} achieved an optimal $O(m+n)$ algorithm for weighted undirected graphs, and extended their techniques to obtain the same bound for weighted DAGs. Their algorithm employs an algebraic framework based on path-enumerating polynomials, where cancellations over a field of characteristic two are used to detect disjoint shortest paths.  


The situation is markedly different for general directed graphs. Here, the first polynomial-time algorithm was given only recently by Bérczi and Kobayashi~\cite{BercziK17}, whose method requires strictly positive edge weights and runs in $O(m^5n)$ time. Thus, despite decades of progress in the undirected setting, the directed case has remained far less understood, with a substantial gap between known upper bounds and what might be achievable.


Our work addresses the central question: \textbf{Can 2-DSP be solved efficiently in general directed weighted graphs?} We give an affirmative answer, presenting new algorithms that both broaden the class of graphs to general weighted graphs without negative or zero weight cycles and significantly improve the running time. Our work builds on the framework of~\cite{AkmalVW24} by employing algebraic techniques  for suitable enumerating polynomials. However, the complexity of path intersections in general directed graphs poses new challenges and we require significantly new ideas to obtain a recursive  decomposition of these polynomials. 


We also introduce and study a natural relaxation of 2-DSP. When no pair of vertex-disjoint shortest paths exists, one may instead ask for two shortest paths whose overlap is as small as possible. This motivates the \emph{Minimum 2-Disjoint Shortest Paths (Min-2-DSP)} problem: given two source–target pairs $(x_1,y_1)$ and $(x_2,y_2)$, find shortest paths $P_1$ and $P_2$ minimizing $|V(P_1) \cap V(P_2)|$, where $V(P)$ denotes the vertex set of path $P$. Our techniques extend naturally to this more general problem, yielding the first polynomial-time algorithms for Min-2-DSP across several graph classes. The problem is inherently more intricate than 2-DSP, since it requires reasoning not only about the existence of disjoint paths but also about quantifying and optimizing the extent of their intersections.



\subsection{Our Contributions}

Our first result settles the complexity of 2-DSP in directed weighted graphs without negative or zero weight cycles, giving the first efficient algorithm in this general setting. Further, it significantly improves the previously best-known $O(m^5n)$-time bound  even for the special case of directed graphs with strictly positive edge weights.  

\begin{restatable}{theorem}{twoDSPdirected}
\label{thm:twoDSPdirected-intro}
The \textsf{2-DSP} problem in directed weighted graphs having no cycles of negative or zero weight is solvable in $O(mn\log n)$ time. 
Moreover, we can report the explicit paths, if they exist, in $O(mn^2\log n)$ time.
\end{restatable}

We next consider the Min-2-DSP problem in directed graphs with positive edge weights. This result provides the first algorithmic framework for minimizing path intersections in the directed setting, showing that even this more general objective remains polynomial-time solvable.
\begin{restatable}{theorem}{mintwoDSPdirected}
\label{thm:mintwoDSPdirected-intro}
The \textsf{Min-2-DSP} problem in directed graphs with positive edge weights is solvable in $O(m^2n^3)$ time.
Moreover, we can report the explicit paths in $O(m^2n^4)$ time bound.
\end{restatable}

For structured graph classes, we obtain significantly faster algorithms. In particular, for DAGs we achieve an optimal linear-time bound. 
\begin{restatable}{theorem}{mintwoDSPDAGs}
\label{thm:mintwoDSPDAGs-intro}
The \textsf{Min-2-DSP} problem in weighted DAGs is solvable in $O(m+n)$ time. Moreover, if we are guaranteed that no disjoint paths exist, then there exists an algorithm which also reports the explicit paths in the same time. 
\end{restatable}

Finally, we extend these results to undirected graphs with positive edge weights, again achieving optimal complexity.  
\begin{restatable}{theorem}{mintwoDSPundirected}
\label{thm:mintwoDSPundirected}
The \textsf{Min-2-DSP} problem in undirected graphs with positive edge weights is solvable in $O(m + n)$ time. 
Moreover, if we are guaranteed that no disjoint paths exist, then there exists an algorithm which also reports the explicit paths in the same time. 
\end{restatable}

\subsection{Related Work}


For $k \geq 3$, the $k$-DSP problem complexity landscape becomes dramatically more challenging, and the complexity landscape diverges sharply between undirected and directed graphs.  In undirected graphs, Lochet~\cite{Lochet21} provided the first polynomial-time algorithm for $k$-DSP (for  constant $k$), though with a running time of $ n^{O(k^{5^k})}$, which was later improved to $n^{O(k \cdot k!)}$ by Bentert et al.~\cite{BentertNRZ21}. It is worth noting that $n^{\Theta(k)}$ is the best we can hope for in undirected graphs, due to $W[1]$-Hardness and  fine-grained lower bounds established in ~\cite{BentertNRZ21, Lochet21, AkmalVW24}. Hence, reducing the exponent from $\Theta(k!)$ to $\Theta(\mathrm{poly}(k))$ remains an important open problem for $k$-DSP in undirected graphs.
In an interesting breakthrough, Pilipczuk, Stamoulis and Wlodarczyk~\cite{PilipczukSW25}  recently claimed an FPT algorithm for $k$-DSP in undirected planar graphs,  providing the first non-trivial graph class to admit such an algorithm.

In contrast, the situation is far less understood for directed graphs with $k \geq 3$: no polynomial-time algorithms are known, and no hardness results have been established either. The complexity of $k$-DSP in directed graphs therefore remains wide open. Nevertheless, $n^{O(k)}$-time algorithms are known for certain restricted classes, including DAGs~\cite{AkmalVW24} and planar digraphs~\cite{BercziK17}.

If zero-weight cycles are permitted, the 2-Disjoint Paths (2-DP) problem becomes a special case of 2-DSP. Since 2-DP is known to be NP-hard in directed graphs~\cite{FortuneHopcroftWyllie80}, excluding such cycles is essential: the assumption that every cycle has strictly positive weight is what ensures tractability of the 2-DSP problem.

A different line of work considers the objective of finding vertex-disjoint paths that minimize the \emph{total length} of the paths. In undirected graphs, Björklund and Husfeldt~\cite{BjorkLundH14} developed an algebraic algorithm for the case of two paths with running time $O(n^{11})$. For $k \geq 3$, the problem remains open even in undirected graphs, and since this objective generalizes $k$-DSP, it inherits all of its hardness barriers~\cite{BentertNRZ21, Lochet21}. More recently, Mari et al.~\cite{MariMPS24} obtained an FPT algorithm for rectangular grid graphs for this objective.


Another natural objective extending the $k$-DSP problem is to maximize the number of source–target pairs that can be connected by vertex-disjoint shortest paths. Clearly this problem inherits the  lower-bound barriers of $k$-DSP, making it natural to ask whether efficient approximation algorithms are possible. However, recent hardness results by Chitnis et al.~\cite{ChitnisTW24} and Bentert et al.~\cite{BentertFG25} rule out even such approximation approaches, underscoring the inherent difficulty of the problem.


To the best of our knowledge, the Min-2-DSP objective has not been studied previously. A related direction appears in the work of Funayama et al.~\cite{FunayamaKU24}, who consider the problem of finding $k$ shortest $s$–$t$ paths that minimize the total number of pairwise vertex intersections among them. We leave the Min-$k$-DSP as an open problem, for $k \geq 3$, for undirected as well as directed graphs.



\subsection{Organization of the Paper}
We introduce the necessary notations, definitions, and algebraic tools in \Cref{sec:preliminaries}. An overview of our techniques is provided in \Cref{sec:overview}. Our $O(mn \log n)$-time algorithm for the 2-DSP problem in general directed graphs is presented in \Cref{sec:2DSP}. The results for the Min-2-DSP problem are discussed separately for directed graphs, DAGs, and undirected graphs in \Cref{sec:Min-2DSP-directed}, \Cref{sec:Min-2DSP-DAGs}, and \Cref{sec:Min-2DSP-undirected}, respectively. Proofs and technical details omitted from the main exposition are included in \Cref{sec:deferred-proofs}.

\section{Preliminaries}
\label{sec:preliminaries}

Let $G = (V, E)$ be a directed or undirected graph. 
For any vertex $v\in V$, let $\IN(v)$ and $\inedges(v)$ denote the set of in-neighbors and in-edges of $v$, respectively, and let $\indeg(v)$ denote the size of these sets. Define $\OUT(v)$, $\outedges(v)$, and $\outdeg(v)$ analogously. 
For any two vertices $x, y\in V$, let $\dist(x, y)$ denote the length of the shortest path from $x$ to $y$ in $G$, and let $\Pi(x, y)$ denote the set of all shortest paths from $x$ to $y$.
We say that a vertex~$w$ is a `\emph{distance-critical vertex}' for the pair $(x, y)$ if its removal increases the distance from $x$ to $y$, i.e., $\dist(x, y, G \setminus \{w\}) > \dist(x, y, G)$. Let $\Delta(x, y)$ denote the set of all distance-critical vertices for pair $(x, y)$, including the endpoints $x$ and~$y$; and define $\delta(x,y)=|\Delta(x,y)|$. Observe that $\delta(x, x) = 1$.
 
Given a path $P$, we denote by $V(P)$ and $E(P)$ the sets of vertices and edges that appear on $P$ respectively. 
Given any two vertices $x,y\in P$, we say $x$ {\bf \em precedes} $y$ on $P$, denoted by $x \preceq_P y$,  if either $x=y$ or $x$ appears before $y$ on $P$. 
For vertices $x, y \in V$ and a path $P$ where $x$ precedes $y$ on $P$, let $P[x, y]$ denote the sub-path of $P$ from $x$ to $y$. Similarly, for a vertex $x \in V$, an edge $e = (a,b) \in E$, and path $P$ where $x$ precedes $e$, let  $P[x,e]$ denote the subpath $P[x,a]$. Define $P[e,x]$ for a suitable $e$, $x$ and $P$ analogously.  
Given two paths $P$ and $Q$, their concatenation is denoted by $P \cdot Q$, provided that the last vertex of $P$ coincides with the first vertex of $Q$.
Given vertices $x$ and $y$, let $V(x,y)$ and $E(x,y)$ denote the sets of all vertices and edges that lie on at least one shortest path between $x$ and $y$ in $G$ respectively. Formally,
$$V(x,y) = \bigcup_{P \in \Pi(x,y)} V(P), \quad \text{and} \quad E(x,y) = \bigcup_{P \in \Pi(x,y)} E(P).$$

\begin{definition}[Internally vertex-disjoint paths]
Let $(x_1, y_1)$ and $(x_2, y_2)$ be two pairs of vertices. A pair of distinct paths, one from $x_1$ to $y_1$ and the other from $x_2$ to $y_2$, are said to be internally vertex-disjoint if they do not share any vertices except those in the set $\{x_1, y_1\} \cap \{x_2, y_2\}$.
\end{definition}

In general, when two directed paths intersect at several vertices, the order in which these vertices appear in the two paths respectively could be different. The following lemma shows that there is more structure when these are shortest paths with a common destination (for proof see appendix).

\begin{restatable}{lemma}{lemorderintersect}
\label{lem:order-intersect}
Let $x_1,x_2, v$ be vertices and $P_1 \in \Pi(x_1,v)$ and $P_2 \in \Pi(x_2,v)$. Suppose $P_1$ and $P_2$ intersect internally and let $S$ be the set of  vertices in $P_1 \cap P_2$. Then the relative order of the vertices in $S$ is the same in both $P_1$ and $P_2$. 
\end{restatable}

\begin{definition}[Twin-crossing paths]
\label{definition: twin-corssing-paths}
Two paths $P_1$ and $P_2$ are said to be {\bf \em twin-crossing} if there exist {\bf \em distinct} vertices $a,b\in V$ such that $a$ precedes $b$ in both paths, and the sub-paths
$P_1[a,b]$ and $P_2[a,b]$ are not identical to each other. Here, $(a,b)$ is referred to as {\bf \em twin-crossing pair} for paths $P_1$ and $P_2$.
\end{definition}

\begin{lemma}[Schwartz-Zippel lemma]
Let $P(x_1, x_2, \dots, x_N)$ be a non-zero polynomial of degree $d$ over a finite field $\mathbb{F}$. If a point is selected uniformly at random from $\mathbb{F}^N$, then the probability that $P$ evaluates to a non-zero value is at least $1-d/{|\mathbb{F}|}$.
\label{lem:SZ}
\end{lemma}

In this paper, we will work with fields of characteristic two. Thus $|\mathbb{F}|$ will be $2^q$, for some integer $q\geqslant 1$. We will be working with $q$ which is $O(\log n)$. This would support addition and multiplication over $\mathbb{F}$ in constant time in the Word-RAM model.

\subsection{Enumerating Polynomials}
For each edge $e = (u,v) \in E$, we introduce an indeterminate $z_{uv}$ (also denoted by $z_e$). For any simple path $P$ consisting of edges $e_1, \ldots, e_\ell$, we associate with it a monomial

$$
f(P) := \prod_{i=1}^{\ell} z_{e_i}.
$$

If $P$ is an empty-path (that is, has zero edges), then $f(P)$ is defined to be $1$.
Note that the monomials corresponding to any two distinct simple paths are distinct. We now proceed to define the concept of enumerating polynomials.

\begin{definition}[Enumerating polynomial]
Given a collection $\C$ of paths, the enumerating polynomial $F_\C$ for $\C$ is defined as
$F_\C := \sum_{P \in \C} f(P).$
Similarly, for a collection $\C$ of pairs of paths, the enumerating polynomial for the collection $\C$ is defined as $F_\C := \sum_{(P_1, P_2) \in \C} f(P_1)f(P_2).$
\end{definition}

For any $x,y\in V$, we use $F(x,y)$ to denote the enumerating polynomial $F_\C$, where $\C = \Pi(x,y)$, i.e., the set of all shortest paths from $x$ to $y$ in $G$. That is,
$$F(x,y) := \sum_{P \in \Pi(x, y)} f(P).$$

\noindent
\cite{AkmalVW24} showed that $F(x,y)$ can be evaluated efficiently, we give a proof sketch (in appendix) for sake of completeness.

\begin{restatable}{lemma}{lemFxy}
\label{lem:Fxy}
Let $G = (V, E)$ be a  directed weighted graph with no negative weight cycles. For any assignment $\{\bar z_e\}_{e \in E}$, the polynomial $F(x, y)$, for all $x,y\in V$, can be evaluated  in $O(mn \log n)$ time.
\end{restatable}

\begin{definition}[Swap operation]  
For any pair of paths $(P_1,P_2) \in \Pi(x_1, y_1) \times \Pi(x_2, y_2)$ 
and any $a,b\in V$ such that $a$ precedes $b$ in both the paths,
we define $\phi_{a,b}(P_1,P_2)$ to be the pair of paths obtained by swapping segments between $a$ and $b$ in $P_1,P_2$. 
That is, if $P_i$ has endpoints $x_i,y_i$, then $Q_i=P_i[x_i,a]\cdot P_{3-i}[a,b]\cdot P_i[b,y_i]$, for $i=1,2$. We refer to $\phi_{a,b}$ as a swap operation at $(a,b)$.
\end{definition}

The following fact is easy to see: 
\begin{claim}
    \label{cl:crossing}
    Let $(P_1, P_2)$ be a pair of paths in $\Pi(x_1, y_1) \times \Pi(x_2, y_2)$ with $(a,b)$ as twin-crossing pair, and let $(Q_1, Q_2) = \phi_{a,b}(P_1,P_2).$ Then the pair $(Q_1, Q_2) \in \Pi(x_1, y_1) \times \Pi(x_2, y_2)$ also has $(a,b)$ as twin-crossing. Further, $\phi_{a,b}(Q_1, Q_2) = (P_1,P_2)$ and the pair $(Q_1,Q_2)$ is distinct from $(P_1,P_2)$.
\end{claim}

\section{Technical Overview}
\label{sec:overview}

In this section, we outline the  technical novelties of our contribution. 
\subsection{2-DSP: From DAGs to Directed Graphs}
\cite{AkmalVW24} studied the $2$-DSP problem for DAGs and undirected graphs, achieving optimal linear-time algorithms. The main idea behind their approach was to compute enumerating polynomials whose monomials are products of edge variables for each choice of vertex-disjoint $(s_1,t_1)$ and $(s_2,t_2)$ shortest paths. The polynomial $F(s_1,t_1)F(s_2,t_2)$ enumerates all shortest path pairs with no disjointness constraints.


For solving $2$-DSP, they showed that it suffices to cancel the effect of intersecting paths from this total enumeration.
Their key insight was that twin-crossing paths contribute zero in a field of characteristic~$2$: if $P_1, P_2$ are twin-crossing at vertices $w, v$, then swapping the segments between $w$ and $v$ results in paths $Q_1, Q_2$ with the same monomial contribution as $P_1, P_2$, and these contributions cancel each other in a characteristic~$2$ field. For the remaining case, namely, the scenario where non-twin crossing paths intersect for the first time at some vertex $v$, they gave techniques that, after precomputing all $F(x, y)$ values via \Cref{lem:Fxy}, require only $O(\indeg(v))$ time per vertex $v$. Summing over all vertices $v$ resulted in the $O(m+n)$ time algorithm.

\medskip

\noindent
\textbf{Challenge in Directed graphs.}
In directed graphs, the main challenge is that paths may intersect in more complicated ways, and there is no universally consistent ``first common intersection'' between $P_1$ and $P_2$ due to the absence of a global topological ordering. 
See \Cref{fig:main-idea-1}.
Unlike DAGs where intersection patterns follow predictable structures, directed shortest paths can intersect and separate multiple times in a complex manner.

\begin{figure}[!ht]
\centering
\includegraphics[height=40mm, page=1]{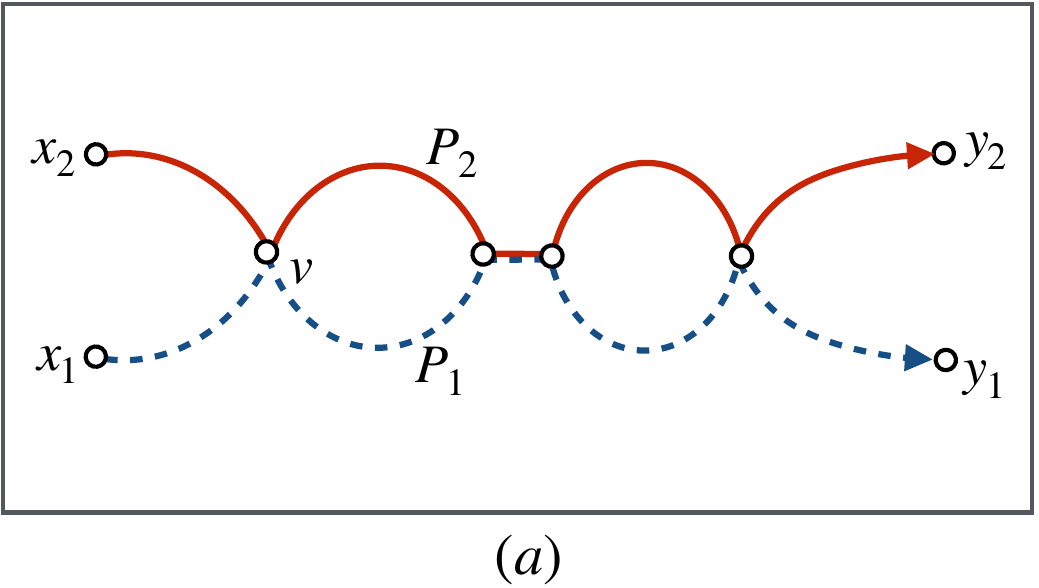}
~\quad~~
\includegraphics[height=40mm, page=2]{main-idea-1.pdf}
\caption{Illustration of interaction of shortest-paths in (a) DAGs, and (b) directed graphs having no cycles of negative or zero weight.}
\label{fig:main-idea-1}
\end{figure}

\medskip

\noindent
\textbf{Our Approach}~ To tackle this complexity, we solve the $2$-DSP problem by studying intersection patterns with respect to the first intersection point on a reference path $P_1 \in \Pi(s_1,t_1)$. Observe that in DAGs, a vertex $v$ is the unique first intersection point for $P_1$ and $P_2$ if and only if
prefixes of $P_1$ and $P_2$ up to $v$ are internally vertex-disjoint. 
In comparison, in general directed graph, a vertex $v$ is the first vertex on $P_1$ that is common between $P_1$ and $P_2$ if and only if the
prefixes of $P_1$ and $P_2$ up to $v$, and the suffix of $P_2$ from~$v$ to $t_2$, {\em all these three segments}, are  internally vertex-disjoint. See \Cref{fig:main-idea-1}.

We define $D_v(s_1,t_1,s_2,t_2)$ to be the enumerating polynomial of pairs $(P_1,P_2)\in \Pi(s_1,t_2)\times \Pi(s_2,t_2)$ such that 
\begin{equation}
P_1[s_1,v],~P_2[s_2,v],~\text{and}~P_2[v,t_2]    
\label{eq:subpaths}
\end{equation}
are internally vertex-disjoint, or equivalently, $v$ is the first vertex on $P_1$ that is common between $P_1,P_2$. Thus, if $F_{\disj}(s_1, t_1, s_2, t_2)$ denotes the enumerating polynomial of all disjoint  pairs of paths $(P_1, P_2) \in \Pi_1(s_1, t_1) \times \Pi_2(s_2,t_2)$, then 
$$
F_{\disj}(s_1, t_1, s_2, t_2)
:=
F(s_1, t_1)F(s_2, t_2) - \sum_v D_v(s_1, t_1, s_2, t_2).
$$ Indeed, the second term in the r.h.s. captures all intersecting pairs $(P_1, P_2) \in \Pi_1(s_1, t_1) \times \Pi_2(s_2,t_2)$ by cycling over all choices of the {\em first} vertex on $P_1$ where this intersection occurs. Thus,  it suffices to check if the above polynomial is non-zero. Now, the disjointness conditions in~\eqref{eq:subpaths} can be relaxed as follows: instead of requiring that $P_1[s_1, v]$ and $P_2[s_2, v]$ are disjoint, we just require that the incoming edges of $v$ in these two paths respectively are distinct. Indeed, if $P_1[s_1, v]$ and $P_2[s_2, v]$ intersect while the incoming edges of $v$ are distinct in these two prefixes, then there exist a common vertex $w \in P_1[s_1, v] \cap P_2[s_2, v]$ such that $(w,v)$ forms a twin crossing pair for $P_1, P_2$. But again, a standard swapping argument shows the contribution of such pairs $(P_1, P_2)$ in the sum $D_v(s_1,t_1,s_2,t_2)$ gets canceled in a field of characteristic $2$.

%

In order to exploit the above observation, we define two additional polynomials.
For any $v$, define $H_v(s_1,t_1,s_2,t_2)$ to be the enumerating polynomial of pairs $(P_1,P_2)\in \Pi(s_1,t_2)\times \Pi(s_2,t_2)$ containing $v$ that satisfy 
$$P_2[s_2,v]~\text{and}~P_2[v,t_2]$$ are internally vertex-disjoint.
Further, for any in-edge $(a,v)$ of $v$, define $H_{av}$ analogously with an additional requirement that the paths $P_1, P_2$ must contain the edge $(a,v)$. The discussion above shows  that $D_v(s_1, t_1, s_2, t_2)$, for any $v$, is given by (see Section~\ref{sec:2DSP} for a more formal proof):
$$
H_v(s_1, t_1, s_2, t_2)
-
\sum_{a \in \IN(v)} H_{av}(s_1, t_1, s_2, t_2).
$$

The definition of $H_v(s_1, t_1, s_2, t_2)$ shows that it can be expressed in terms of $F_{\disj}(s_1, v, v, t_2)$, and similarly for $H_{av}(s_1, t_1, s_2, t_2)$ (see Lemma~\ref{lem:Hv} and~\ref{lem:Hav}).
Thus, we are able to recursively express $F_{\disj}(s_1, t_1, s_2, t_2)$ in terms of polynomials $F_{
\disj}$ over a different set of quadruples. Finally, we show that overall only $O(m)$ quadruples are encountered in such manner, and their computation (i.e., evaluation over randomly chosen values) can be done efficiently using dynamic programming.

\subsection{Solving Min-2-DSP in general directed graphs}

The Minimum-2-Disjoint Shortest Paths (Min-2-DSP) problem is significantly more challenging than the standard 2-DSP problem. Let $(s_1, t_1)$ and $(s_2, t_2)$ be the input source-destination pairs, and let $k^*$ denote the minimum possible overlap between any pair $(P_1, P_2) \in \Pi_1(s_1, t_1) \times \Pi_2(s_2, t_2)$ (we call such a pair an {\em optimal pair
}). 

One could define a suitable enumerating polynomial of all such optimal pairs. 
The first challenge in computing $k^*$ arises when every optimal pair $(P_1, P_2)$ is twin-crossing. Indeed, 
the standard swap operation would result in another optimal pair $(Q_1, Q_2)$ such that the contributions of $(P_1, P_2)$ and $(Q_1, Q_2)$ would cancel out in the enumerating polynomial. While this cancellation property was crucial in our  algorithm for the standard 2-DSP problem, it inadvertently causes the cancellation of contributions from optimal solutions in the Min-2-DSP setting.
To overcome this limitation, we introduce the notion of a concordant pair.

\begin{definition}[Concordant pair]
Let $P_1$ and $P_2$ be two paths with endpoints $x_1, y_1$ and $x_2, y_2$, respectively. A pair of vertices $(a, b)$, not contained in the set $\{x_1, y_1\} \cap \{x_2, y_2\}$, is said to be a concordant pair for $P_1$ and $P_2$ if $a$ appears before $b$ on $P_1$ and $P_2$; and 
the subpaths $P_1[x_1,a],P_2[x_2,a]$ are internally 
vertex-disjoint, and likewise $P_1[b,y_1],P_2[b,y_2]$ are internally 
vertex-disjoint.
\end{definition}

Note that for general directed graphs,  then there can be multiple concordant pairs with respect to any pair of intersecting paths (see~\Cref{fig:main-idea-2}). We denote by $\C(P_1,P_2)$ the set of all concordant pairs with respect to paths $P_1,P_2$.




\begin{figure}[!ht]
\centering
\includegraphics[height=40mm, page=1]{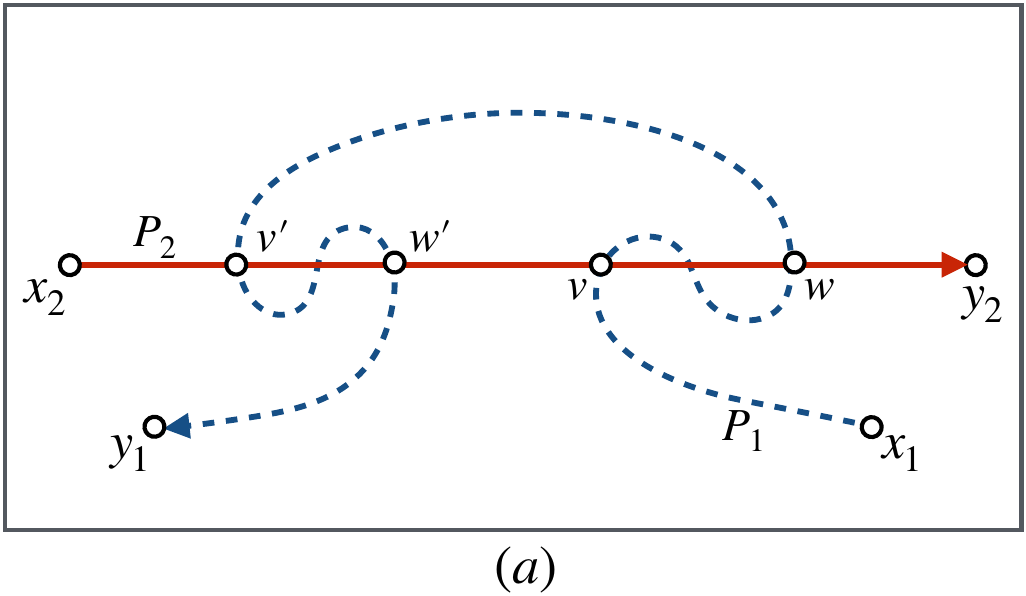}
\quad~
\includegraphics[height=40mm, page=2]{main-idea-2.pdf}
\caption{Illustration of paths $P_1,P_2$ with two concordant pairs. Observe that $\C(P_1,P_2)=\C(P_1',P_2)$. 
}
\label{fig:main-idea-2}
\end{figure}

\begin{restatable}{lemma}{lemconcordantpairs}
\label{lemma:concordant-pairs-1}
Let $(P_1,P_2)\in \Pi(s_1,t_1)\times \Pi(s_2,t_2)$ be a pair of paths, and $(v_1,w_1),\ldots,(v_r,w_r)$ be the concordant pairs with respect to $P_1,P_2$. Then, the following holds.

\begin{enumerate}


\item $($\emph{Vertex-disjointness}$)$ 
For $i\neq j$, 
$V(P_1[v_i,w_i]) \cap V(P_1[v_j,w_j])$ and 
$V(P_2[v_i,w_i]) \cap V(P_2[v_j,w_j])$ are empty sets; 
that is, the subpaths corresponding to concordant pairs on $P_1$ (resp. $P_2$) are vertex-disjoint.

\item $($\emph{Order Reversal}$)$ 
If the concordant pairs are ordered along $P_1$ as   $(v_1,w_1),\ldots,(v_r,w_r)$, then on $P_2$ this ordering is exactly reversed.

\end{enumerate}

\end{restatable}

Intuitively, the concordant pairs capture the structural overlap for any pair of paths having minimum overlap.  In order to explain this we present the following observation (for proof see appendix).

\begin{restatable}{observation}{obspathCritical}
\label{observation:path-intesection-critical-nodes}
For any $v,w\in V$, there exist two shortest paths from $v$ to $w$ in $G$ that intersect only at $(v,w)$-distance critical vertices.
\end{restatable}

As a corollary we get the following.

\begin{lemma}
\label{lemma:concordant-pairs-2}
For any $(s_1,t_1)$ and $(s_2,t_2)$ shortest paths $P_1,P_2$ intersecting at exactly $k^*$ internal vertices and any $(u,v)\in \C(P_1,P_2)$, the subpaths of $P_1,P_2$ from $u$ to $v$ have $\delta(u,v)$ vertices in common. Furthermore, these common vertices are precisely the $(u,v)$-distance-critical vertices.
\end{lemma}

This suggests a shift in strategy: rather than directly searching for an optimal path pair, we can search for a pair of paths such that the total count of distance-critical vertices across all concordant pairs is minimized. This leads us to the following definition.

\begin{definition}
The interaction complexity of two paths $P_1$ and $P_2$ is defined as
$$
\gamma(P_1, P_2) = \sum_{(v, w) \in \C(P_1, P_2)} \delta(v, w).
$$
\end{definition}

Note that if $P_1,P_2$ is an optimal solution to Min-2-DSP then $\gamma(P_1, P_2)=k^*$. We define $F_{\disj, k}(s_1,t_1,s_2,t_2)$ as the enumerating polynomial of all pair of paths in $\Pi(s_1,t_1)\times \Pi(s_2,t_2)$ that have at most $k$ interaction complexity.
Our goal is to return the smallest $k$ for which $F_{\disj, k}(s_1,t_1,s_2,t_2)$ is non-zero.

In order to compute $F_{\disj, k}$, we introduce several helper polynomials and present relations between them. These are more complicated in nature, so we defer the discussion to \Cref{sec:Min-2DSP-directed}.

\subsection{Solving Min-2-DSP in DAGs and Undirected graphs} 
For DAGs and undirected graphs, the Min-2-DSP problem admits a simpler and more efficient $O(m+n)$ time algorithm. 
First assume that given pairs $(s_1, t_1)$ and $(s_2, t_2)$, the desired minimum number of intersections $k^*$ between any pair of paths $(P_1, P_2) \in \Pi(s_1, t_1) \times \Pi(s_2,t_2)$ is strictly positive (otherwise, we can employ the results in~\cite{AkmalVW24}). For a vertex $v$, let $\alpha(v)$ denote the minimum number of intersection of any pair of paths $(P_1, P_2) \in \Pi(s_1, t_1) \times \Pi(s_2, t_2)$ such that $v$ lies on both of these paths. Since the optimal pair of paths must intersect at some vertex, it is not hard to show that $k^*$ is equal to $\min_v \alpha(v)$ (if $v$ does not lie on any pair of paths in  $\Pi(s_1, t_1) \times \Pi(s_2, t_2)$, we can assume $\alpha(v)$ to be infinity).


A naive method for calculating $\alpha(v)$ would be following: (i) find the $(v, t_1), (v, t_2)$ shortest paths with minimum number of intersections (call this number $\alpha_1(v)$) and (ii) find the $(s_1, v), (s_2, v)$ shortest paths with minimum number of intersections (call this number $\alpha_2(v)$). 
If $G$ is DAG, or if $G$ is undirected and $k^*\geqslant 2$, then it can be seen that $\alpha(v) = \alpha_1(v) + \alpha_2(v)$. 
For each vertex $v$, $\alpha_1(v)$ and $\alpha_2(v)$ can be computed in $O(m+n)$ time by building dominator trees~\cite{LT79, GeorgiadisT05} in appropriate subgraphs of $G$. This results in overall running time of $O(mn)$.

While it may seem that computing $\alpha(v)$ for any fixed $v$ takes $O(m)$ time, our main insight is that this computation can be done more efficiently, such that computing $\alpha(v)$ over all $v \in V$ takes $O(m + n)$ total time. For doing so, we take inspiration from the notion of {\em bicoloured components} defined by Lochet~\cite{Lochet21} and partitioning the vertex set $V$ into smaller components where Min-2-DSP reduces to Min-2-DP.

\section{2-DSP in Directed graphs}
\label{sec:2DSP}
For any vertex pairs $(x_1, y_1)$ and $(x_2, y_2)$, let $F_{\disj}(x_1, y_1, x_2, y_2)$ be the enumerating polynomial of the collection of all pairs $(P_1, P_2) \in \Pi(x_1, y_1) \times \Pi(x_2,y_2)$ such that  $P_1$ and $P_2$ are internally vertex-disjoint. 



Observe that the polynomial $F_{\disj}(x_1, y_1, x_2, y_2)$ is non-zero if and only if there exist two shortest paths, one from $x_1$ to $y_1$ and the other from $x_2$ to $y_2$, that are internally vertex-disjoint. This follows from the fact that each such pair contributes a distinct monomial to $F_{\disj}(x_1, y_1, x_2, y_2)$, determined by the set of edges used along both paths. Indeed, if $(P_1, P_2)$ is a pair of internally vertex-disjoint shortest paths for $(x_1, y_1)$ and $(x_2, y_2)$, then for each $i=1,2$, the path from $x_i$ to $y_i$ remains unique in the subgraph $G' = (V, E(P_1) \cup E(P_2))$, unless the two vertex pairs coincide.

\subsection{Helper polynomials}
In this subsection, we will provide construction of various helper polynomials 
that will be crucial in solving 2-DSP in directed graphs.
The helper polynomials introduced in this section will be helpful in providing a recursive formulation for the polynomial $F_{\disj}(x_1, y_1, x_2, y_2)$.

\subsubsection*{I. Disjoint Prefix-Suffix linkages up to a Shared Vertex/Edge}

For any $v\in V$, let $H_v(x_1, y_1, x_2, y_2)$ be the enumerating polynomial of all pairs $(P_1, P_2) \in \Pi(x_1,y_1) \times \Pi(x_2,y_2)$  such that
\begin{enumerate}
\item $v\in V(P_1)\cap V(P_2)$, and 
\item Prefix $P_1[x_1,v]$ and suffix $P_2[v,y_2]$ are internally vertex-disjoint.
\end{enumerate}
Similarly, for any edge $e=(a,v)\in E$, let $H_{av}(x_1, y_1, x_2, y_2)$ be the enumerating polynomial of all  all pairs $(P_1, P_2) \in \Pi(x_1,y_1) \times \Pi(x_2,y_2)$  such that
\begin{enumerate}
\item $e\in E(P_1)\cap E(P_2)$, and 
\item Prefix $P_1[x_1,e]$ and suffix $P_2[e,y_2]$ are internally vertex-disjoint.
\end{enumerate}

\smallskip

The next two lemmas (with proofs in appendix) present relations between polynomials $H_v, H_{av}$, and $F_{\disj}$.

\smallskip

\begin{restatable}{lemma}{lemHv}
\label{lem:Hv}
For any $v \in V$ satisfying $\dist(x_i, y_i) = \dist(x_i, v) + \dist(v, y_i)$, for $i=1,2$, we have 
$$H_v(x_1, y_1, x_2, y_2) = F(x_2, v) \cdot F_{\disj}(x_1, v, v, y_2) \cdot F(v, y_1).$$
\end{restatable}

\begin{restatable}{lemma}{lemHav}
\label{lem:Hav}
For any $e = (a, v) \in E$, satisfying $\dist(x_i, y_i) = \dist(x_i, a) + wt(a, v) + \dist(v, y_i)$, for $i = 1, 2$, we have
$$H_{av}(x_1, y_1, x_2, y_2) = F(x_2, a) \cdot z_{av}^2 \cdot F_{\disj}(x_1, a, v, y_2) \cdot F(v, y_1).$$
\end{restatable}

\subsubsection*{II. Paths with First Intersection at a Given Vertex along a Reference Path}
We define $D_v(x_1, y_1, x_2, y_2)$ to be the enumerating polynomial of the collection of pairs $(P_1, P_2) \in \Pi(x_1, y_1) \times \Pi(x_2,y_2)$ whose first intersection point with respect to $P_1$ is vertex~$v$%
\footnote{Note that $v$ may not be the first/last intersection point with respect to path $P_2$.}.
More formally, the paths $P_1$ and $P_2$ satisfy the following.
\begin{enumerate}
\item $v\in V(P_1)\cap V(P_2)$, and 
\item Prefix $P_1[x_1,v]$ is internally vertex-disjoint with $P_2[x_2,v]$ as well as $P_2[v,y_2]$.
\end{enumerate}

\smallskip

It is important to observe that if $v \notin V(x_1, y_1) \cap V(x_2, y_2)$, then $D_v(x_1, y_1, x_2, y_2)$ will be a zero polynomial.
We next present a relation between polynomials $H_v, H_{av},$ and $D_v$.
The lemma below crucially uses the fact that the underlying field has  characteristic two (for proof see appendix).

\begin{restatable}{lemma}{lemDv}
\label{lem:Dv}
For any $v \in V$ satisfying $\dist(x_i, y_i) = \dist(x_i, v) + \dist(v, y_i)$ for $i = 1, 2$, we have
$$
D_v(x_1, y_1, x_2, y_2)
~=~
H_v(x_1, y_1, x_2, y_2)
-
\sum_{a \in \IN(v)} H_{av}(x_1, y_1, x_2, y_2).
$$
\end{restatable}

\begin{figure}[!ht]
\centering
\includegraphics[width=0.9\textwidth]{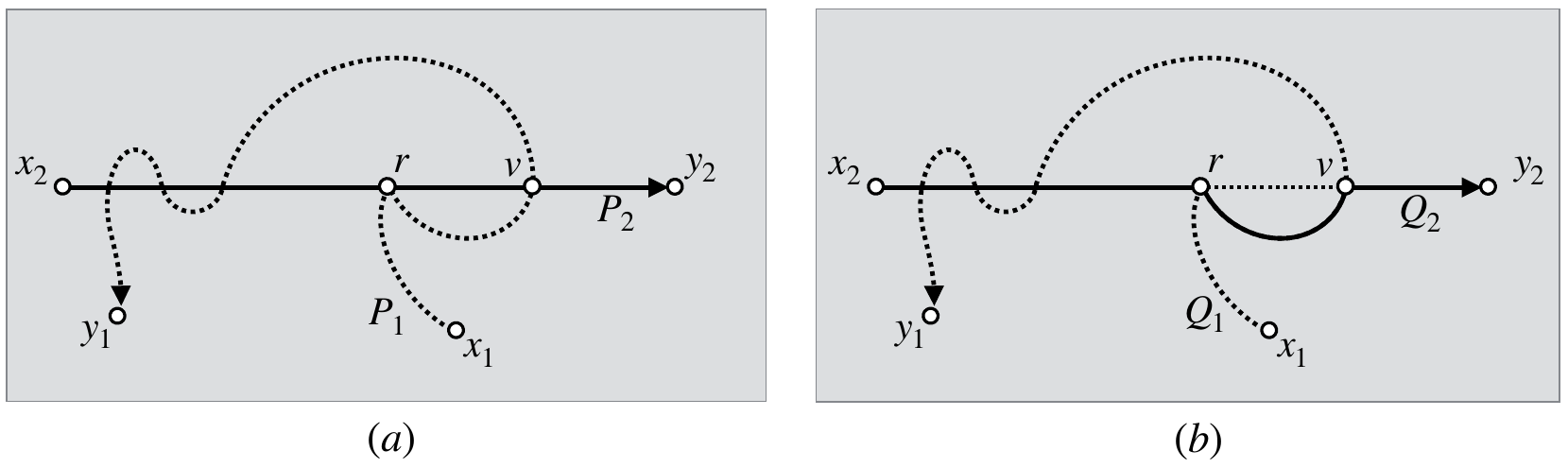}
\caption{Illustration of the swap operation on pair $(P_1,P_2)\in \A\setminus \B$.}
\label{figure:swap-2-DSP}
\end{figure}

\subsubsection*{III. Intersecting and Disjoint Shortest paths}
We define $F_{\cap}(x_1, y_1, x_2, y_2)$ to be the enumerating polynomial for all pairs of shortest paths between pairs $(x_1, y_1)$ and $(x_2, y_2)$ that are not internally vertex-disjoint.
Observe that  $$F_{\cap}(x_1, y_1, x_2, y_2) = \sum_{v \in V \setminus (\{x_1,y_1\}\cap\{x_2,y_2\})} D_v(x_1, y_1, x_2, y_2).$$ Indeed,  if two paths $P_1, P_2$ are intersecting (at an internal vertex), there must exist a unique earliest intersection point $v$ on $P_1$ (that lies outside the set $\{x_1,y_1\}\cap\{x_2,y_2\}$).
%
%
%
%
Since every pair of $(x_1,y_1)$ and $(x_2,y_2)$ shortest paths
is either internally vertex-disjoint or else intersects at some common vertex outside 
the set $\{x_1,y_1\}\cap\{x_2,y_2\}$, we obtain the following result.

\begin{lemma}
\label{lem:F-disj}
The enumerating polynomial $F_{\disj}(x_1, y_1, x_2, y_2)$ is given by
$$F_{\disj}(x_1, y_1, x_2, y_2) ~=~
F(x_1,y_1)F(x_2,y_2)
-\sum_{v \in V \setminus (\{x_1,y_1\}\cap\{x_2,y_2\})} D_v(x_1, y_1, x_2, y_2).$$
\end{lemma}


\subsection{Solving 2-DSP in ${O(m^2\log n)}$ time}
In this subsection, we present an algorithm for solving the 2-DSP problem in weighted directed graphs that takes $O(m^2)$ time. Let $(s_1, t_1)$ and $(s_2, t_2)$ be the input terminal pairs. By~Schwartz-Zippel lemma (\Cref{lem:SZ}), in order to check if $F_{\disj}(s_1, t_1, s_2, t_2)$ is non-zero it suffices to evaluate this polynomial  on a random assignment of the variables $\{z_e\}_{e \in E}$ drawn from a sufficiently large finite field $\mathbb{F}$. The algorithm for evaluating $F_{\disj}(s_1, t_1, s_2, t_2)$ at these random values is given in Algorithm~\ref{algo:2-DSP-directed}.

\smallskip

We begin by assigning random values $\barz_e$ from $\mathbb{F}$ to each of the edge variables $z_e$ respectively. For a polynomial $R$ involving the variables $z_e, e \in E$, we shall use ${\bar R}$ to denote the value in $\mathbb{F}$ obtained by evaluating $R$ at $\barz_e, e \in E$.  Using \Cref{lem:Fxy}, we compute $\barF(x, y)$ for all pairs of vertices in $O(mn \log n)$ time. In order to compute $\barFd(s_1,t_1, s_2, t_2),$ we would need to compute $\barFd(x_1,y_1, x_2, y_2)$ for several tuples $(x_1, y_1, x_2, y_2)$. We store all such tuples in a list $\cL$. In particular, $\cL$ contains the following: 
%
(i) tuples $(s_1, v, v, t_2)$, for each $v \in V$;
(ii) tuples $(s_1, a, v, t_2)$, for each  $(a,v) \in E$; and (iii) the query tuple $(s_1, t_1, s_2, t_2)$.

\smallskip



Now we evaluate $\barFd(x_1,y_1, x_2, y_2)$ for each such tuple in $\cL$ using dynamic programming. 
For this, we define a partial order $\prec$ on the set of vertices $V$ of the graph $G$ as follows: for any distinct $x, y \in V$, we say $x \prec y$ if there exists an $s_1$-to-$y$ shortest path in $G$ that contains $x$ (as an internal vertex). This is well-defined, as $G$ contains no cycles of zero weight.

\begin{definition}
Given the partial order $\prec$ on $V$, let $\cL^*$ denote an ordered list such that tuples of the form $(s_1, b, u, t_2)$, $(s_1, u, u, t_2)$ precede tuples like $(s_1, a, v, t_2)$, $(s_1, v, v, t_2)$ whenever $u \prec v$, i.e., whenever $u$ lies on an $(s_1,v)$ shortest path. The tuple $(s_1,t_1,s_2,t_2)$ appears at the end of the list.
\end{definition}

Computing $F_{\disj}$ for tuples in the order given  by $\cL^*$ ensures that 
 that while  evaluating $\barFd(x_1,y_1, x_2, y_2)$, all the values that are required for this computation are available. 
In the base case, when $x_1 = y_1$ and $x_2 = y_2$, we set $\barFd(x_1,y_1, x_2, y_2)=1$.




\bigskip

To evaluate $\barFd(x_1, y_1, x_2, y_2)$, we use the identity (see~\Cref{lem:F-disj}):
$$\barFd(x_1, y_1, x_2, y_2) = \barF(x_1, y_1) \barF(x_2, y_2) ~- \sum_{v \in V\setminus (\{x_1,y_1\}\cap \{x_2,y_2\})} \barD_v(x_1, y_1, x_2, y_2).$$

By \Cref{lem:Dv} ,$\barD_v$ can be expressed as

$$\barD_v(x_1, y_1, x_2, y_2) = \barH_v(x_1, y_1, x_2, y_2) ~- \sum_{(a,v) \in E} \barH_{av}(x_1, y_1, x_2, y_2),$$

for each $v\in V(x_1,y_1)\cap (x_2,y_2)$.
Now, by \Cref{lem:Hv}, we have

$$\barH_v(x_1, y_1, x_2, y_2) = \barF(x_2, v) \cdot \barFd(x_1, v, v, y_2) \cdot \barF(v, y_1).$$
We shall show in Lemma~\ref{lem:lstar} that the ordering induced by $\cL^*$ ensures that the r.h.s. above has already been computed by the dynamic program. 
Similarly, by \Cref{lem:Hav}, for any edge $(a,v)$,
$$\barH_{av}(x_1, y_1, x_2, y_2) = \barF(x_2, a) \cdot z_{av}^2 \cdot \barFd(x_1, a, v, y_2) \cdot \barF(v, y_1).$$

Finally, we output that there exists two disjoint shortest paths between $(s_1,t_1)$ and $(s_2,t_2)$ if and only if $\barFd(s_1, t_1, s_2, t_2)$ evaluates to non-zero.

\medskip



\begin{restatable}{lemma}{lemlstar}
\label{lem:lstar}
We can compute in $O(mn)$ time the ordered list $\cL^*$ so that, when invoking pairs according to this ordering in our dynamic program, all subproblems required are already computed.
\end{restatable}

\paragraph{Correctness and Running Time} 
The correctness of the procedure follows from the recursive expressions in~\Cref{lem:Hv}, \Cref{lem:Hav}, \Cref{lem:Dv}, and \Cref{lem:F-disj}. 
We now analyze the running time.
First observe that $\cal L$ has $O(m)$ tuples: there are at most $n$ choices for the tuples $(s_1, v, v, t_2)$ and at most $m$ choices for $(s_1, a, v, t_2)$. The description of the algorithm shows that 
 we can evaluate $\barD_v(x_1, y_1, x_2, y_2)$ in $O(\indeg(v))$ time for any fixed $v$, and hence we can compute $\barFd(x_1, y_1, x_2, y_2)$ in $O(m)$ total time per quadruple.  Since we evaluate these polynomials for $O(m)$ quadruples in $\cL$, the overall running time is bounded by $O(m^2)$, after an initial preprocessing step of $O(mn \log n)$. This yields a randomized $O(m^2 \log n)$-time algorithm for 2-DSP with high success probability.

\begin{algorithm}[!ht]
\setstretch{1.25}
Assign each edge $(u,v) \in E$ an independent and uniformly random value $\barz_{uv}$ from $\mathbb{F}$\;
Use Lemma~\ref{lem:Fxy} to compute $\barF(x,y)$ for all $x, y \in V$\;
Compute the ordered list $\cL^*$ using \Cref{lem:lstar}\;
\ForEach{$(x_1, y_1, x_2, y_2)$ in the ordered list $\cL^*$}{
  %
  \lIf{$x_1=y_1$ \textup{and} $x_2=y_2$}{Set $\barFd(x_1,y_1, x_2, y_2)=1$ and \textbf{continue} for-loop}
  \ForEach{$v \in V \setminus (\{x_1,y_1\}\cap\{x_2,y_2\})$}{
    \uIf{$v\in V(x_1, y_1) \cap V(x_2, y_2)$}{Compute $\barD_v(x_1,y_1,x_2,y_2)$ as
    $$\hspace{-2.5mm}
    \barF(x_2,v) \cdot \barFd(x_1, v, v, y_2) \cdot \barF(v, y_1)~~~
    - 
    \hspace{-2mm}
    \sum_{\substack{(a,v) ~\in \\ E(x_1,y_1),\\  E(x_2,y_2)}}
    \hspace{-1mm}
    z_{av}^2 \cdot F(x_2,a) \cdot \barFd(x_1,a, v, y_2) \cdot \barF(v, y_1);
    $$
    \vspace{-4mm} \label{l:Dv}
    }
    \Else{Set $\barD_v(x_1,y_1,x_2,y_2)=0;$}
  }
  Compute
  $$
  \barFd(x_1, y_1, x_2, y_2) = \barF(x_1, y_1) \cdot \barF(x_2, y_2)
  -
  \sum_{\substack{v \in V \setminus (\{x_1,y_1\}\cap\{x_2,y_2\})}} \barD_v(x_1, y_1, x_2, y_2);
  $$
}
{\bf Return} True if $\barFd(s_1, t_1, s_2, t_2) \neq 0$, else False\;
\caption{\textsc{2-DSP-Directed}($G, s_1, t_1, s_2, t_2$)}
\label{algo:2-DSP-directed}
\end{algorithm}


\subsection{An $O(mn \log n)$ time algorithm for 2-DSP}
We now present an improved algorithm for solving the 2-DSP problem in $O(mn \log n)$ time. Recall from the previous subsection that Algorithm~\ref{algo:2-DSP-directed} computes $F_{\disj}(s_1, t_1, s_2, t_2)$ by evaluating $F_{\disj}$ for all quadruples in a list $\cL$ of size $O(m)$, where each evaluation incurs $O(m)$ time. Note that we can afford $O(m)$ time computation for evaluating $\barFd(s_1, v, v, t_2)$ for each $v$ --  since there are only $O(n)$ such vertices, this would incur a total of $O(mn)$ time. Therefore, our focus will be on improving the computation of $\barFd(s_1, a, v, t_2)$, where  $(a,v)$ is an edge in $G$.




Consider a vertex $v \in V(s_1, t_1) \cap V(s_2, t_2)$. We show how to compute $\barFd(s_1,a,v,t_2)$, for all relevant edges $(a,v) \in \inedges(v)$  in  $O(m)$ total time.

\begin{figure}[!ht]
    \centering
    \includegraphics[width=0.9\textwidth]{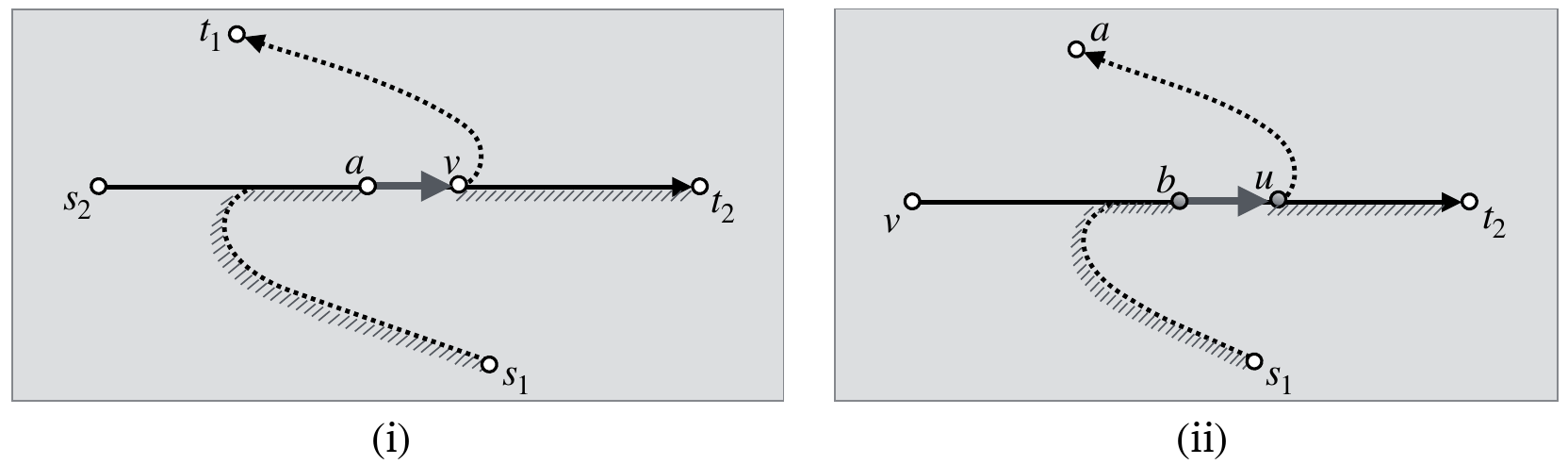}
    \caption{Illustration of interaction between shortest paths contributing to the polynomials (i)~$H_{av}(s_1,t_1,s_2,t_2)$, and (ii) its building block $H_{bu}(s_1, a, v, t_2)$.
    (Note that only the shaded prefix/suffix segments need to be vertex-disjoint.)}
    \label{figure:algo-2DSP}
\end{figure}

Let $a \in \IN(v)$ be such that $(a,v) \in E(s_1, t_1) \cap E(s_2, t_2)$. Now,~\Cref{lem:F-disj} shows that computation of $\barFd(s_1, a, v, t_2)$ requires 
us to evaluate $\barD_u(s_1,a,v,t_2)$ for each $u \in V(s_1, a) \cap V(v, t_2)$ (see  \Cref{figure:algo-2DSP}). Now, line~\ref{l:Dv} in Algorithm~\ref{algo:2-DSP-directed} shows that 

\begin{align*}
\barD_u(s_1,a,v,t_2) & = \barF(v,u) \cdot \barFd(s_1, u, u, t_2) \cdot \barF(u, a) -  \sum_{\substack{(b,u) ~\in \\ E(s_1,a),\\  E(v,t_2)}}
    z_{bu}^2 \cdot \barF(v,b) \cdot \barFd(s_1,b, u, t_2) \cdot \barF(u, a) \\[4mm]
& = \barF(v,u) \cdot \barFd(s_1, u, u, t_2) \cdot \barF(u, a) 
    -  
    \underbrace{
    \bigg(
    \sum_{\substack{(b,u) ~\in \\ E(s_1,\pmb{u}),\\  E(v,t_2)}}
    z_{bu}^2 \cdot \barF(v,b) \cdot \barFd(s_1,b, u, t_2) 
    \bigg)}_{\text{independent\, of choice of edge } (a,v)}
    \cdot \barF(u,a)~.\\
\end{align*}

The last inequality holds because $\inedges(u)\cap E(s_1,a)=\inedges(u)\cap E(s_1,u)$, for any $u\in V(s_1,a)$.
Specifically, 
(i) if $(b, u) \in E(s_1, a)$, then $(b, u) \in E(s_1, u)$ by the substructure property of shortest paths;
(ii) conversely, if $(b, u) \in E(s_1, u)$ and $u$ lies on an $(s_1, a)$-shortest path $P$, then $(b, u) \in E(s_1, a)$, since the prefix of $P$ up to $u$ can be replaced by an $(s_1, u)$-shortest path passing through $(b, u)$.
We emphasize that the sum in the parenthesis above depends only on $u$ and $v$, and not on the specific choice of $a \in \IN(v)$. Hence, we can pre-evaluate this expression once per $u$ in $O(\indeg(u))$ time and reuse it for all relevant in-neighbors $a$ of $v$. Since $\sum_{u} \indeg(u) = O(m)$, the total time across all such computations for a fixed vertex $v$ is $O(m)$.

Summing over all $v \in V(s_1, t_1) \cap V(s_2, t_2)$ gives a total time of $O(mn)$ to evaluate all required instances of $F_{\disj}$.

Finally, using \Cref{lem:Fxy}, we can evaluate all $\barF(x, y)$ values in $O(mn \log n)$ time. Hence, the overall running time of our improved algorithm for detecting two disjoint shortest paths becomes $O(mn \log n)$. Thus, we can conclude with the following theorem.

\begin{theorem}
\label{thm:twoDSPdirected}
The \textsf{2-DSP} problem in directed weighted graphs having no cycles of negative or zero weight is solvable in $O(mn\log n)$ time. 
\end{theorem}

\subsection{Reporting the Paths}
\label{sec:reportingpaths2DSP}
We now show how to compute a pair of internally vertex-disjoint shortest paths from $s_1$ to $t_1$ and from $s_2$ to $t_2$ in $O(mn^{2}\log n)$-time, if such paths exist.

Recall that the proof of \Cref{thm:twoDSPdirected} provides an $O(mn\log n)$ time algorithm for solving the $2$-DSP problem for a collection of $O(m)$ quadruples lying in a list $\cL^*$. In order to report the paths, we show how to compute an out-neighbor $w^*$ of $s_2$ that appears in some solution to the $2$-DSP problem in the same time bound.
Recall that the list $\cL^*$ comprises of (i)~tuples $(s_1, v, v, t_2)$, for each $v \in V$;
(ii)~tuples $(s_1, a, v, t_2)$, for each  $(a,v) \in E$; and (iii)~the query tuple $(s_1, t_1, s_2, t_2)$.

To compute the out-neighbor $w^*$  of $s_2$, we remove the tuple $(s_1, t_1, s_2, t_2)$ from list $\cL^*$, and append the tuples $(s_1, t_1, w, t_2)$, for $w \in \OUT(s_2)\cap V(s_2, t_2)$.
%
%
Now, we can solve the $2$-DSP problem for all pairs in $\cL^*$ in the graph $G\setminus \{s_2\}$, in the same time bound of $O(mn\log n)$.
By Schwartz-Zippel lemma (\Cref{lem:SZ}), with high probability, the out-neighbors of $s_2$ that participates in solution to the $2$-DSP problem are precisely those $w^* \in \OUT(s_2)\cap V(s_2, t_2)$ for which $\barF_{\disj}(s_1, t_1, w^*, t_2)$ evaluates to non-zero in the graph $G\setminus \{s_2\}$. 
After finding the node $w^*$, we {\bf delete} vertex $s_2$ from $G$, and consider a smaller instance of the $2$-DSP problem, where source $s_2$ is replaced with $w^*$.
This process is repeated at most $O(n)$ times, helping us recover an $(s_2, t_2)$-shortest path $P_2$, disjoint from some $(s_1, t_1)$-shortest path in $G$. 

Finally, we delete the vertices of $P_2$ from $G$, and find an $(s_1, t_1)$-shortest path $P_1$ in the resulting graph in $O(mn)$ time using the Bellman-Ford algorithm.
This provides us an algorithm for reporting a pair of disjoint $(s_1,t_1)$ and $(s_2,t_2)$ shortest paths, if they exists, in $O(mn^2 \log n)$ time.

\Cref{thm:twoDSPdirected} together with discussion above proves the following result.

\twoDSPdirected*

\section{Min-2-DSP in Directed graphs}
\label{sec:Min-2DSP-directed}

In this section, we consider the  \textsf{Min-2-DSP} problem in directed graphs with positive edge weights and show that it can be  solved in $O(m^2n^3)$ time.
For any vertex pairs $(x_1, y_1)$ and $(x_2, y_2)$, and any integer $k$, let
$$
F_{\disj,k}(x_1, y_1, x_2, y_2)
$$
be the enumerating polynomial of all pairs of shortest paths $P_1$ from $x_1$ to $y_1$ and $P_2$ from $x_2$ to $y_2$ such that the \emph{interaction complexity}
$$
\gamma(P_1, P_2) = \sum_{(v, w) \in \C(P_1, P_2)} \delta(v,w)
$$
is at most $k$. In order to solve Min-2-DSP it suffices to determine the smallest integer~$k$ for which polynomial $F_{\disj,k}(s_1, t_1, s_2, t_2)$ is non-zero, where $(s_1,t_1)$ and $(s_2,t_2)$ are input source-destination pairs. 

\smallskip

For any vertex pairs $(x_1, y_1)$ and $(x_2, y_2)$, any integer $k$, and vertices $v,w\in V$, define $$D_{v,w,k}(x_1, y_1, x_2, y_2)$$ to be the enumerating polynomial of all pairs $(P_1,P_2) \in \Pi(x_1, y_1) \times \Pi(x_2, y_2)$, having interaction complexity $k$ and such that $(v, w)$ is the \emph{first (with respect to $P_1$)} concordant pair in $\C(P_1, P_2)$.


\smallskip
The following lemma provides relation between polynomials $F_{\disj,k}$ and $D_{v,w,k}$.

\begin{lemma}
\label{lem:min2dsp-id0}
Let   $(x_1, y_1)$ and $(x_2, y_2)$ be pairs satisfying the property that  any pair of paths $(P_1, P_2) \in \Pi(x_1, y_1) \times \Pi(x_2, y_2)$ intersect internally. Then, for any  
integer $k\geqslant 1$, we have
$$
F_{\disj,k}(x_1,y_1,x_2,y_2)
=
\sum_{\substack{v,w \in V
}}
D_{v,w,k}(x_1, y_1, x_2, y_2)
$$
\end{lemma}

\begin{proof}
Let $(P_1, P_2) \in \Pi(x_1, y_1) \times \Pi(x_2, y_2)$ be a pair of paths with interaction complexity $k$. Furthermore, let $(v, w) \in \C(P_1, P_2)$ be the first concordant pair appearing on $P_1$. Then, the enumerating polynomial for such pairs is given by $D_{v, w, k}(x_1, y_1, x_2, y_2)$. To obtain $F_{\disj, k}(x_1, y_1, x_2, y_2)$, we sum over all choices of concordant pairs.
This proves the desired claim.
\end{proof}


To efficiently determine whether $D_{v,w,k}(x_1, y_1, x_2, y_2)$ is a non-zero polynomial, we define next some helper polynomials.

\subsection{Additional Helper Polynomials}
For any vertex pairs $(x_1, y_1)$ and $(x_2, y_2)$, any integer $k$, and any $v,w\in V$, define

\[
H_{v,w,k}(x_1, y_1, x_2, y_2)
\]
to be the enumerating polynomial  for all pairs of shortest paths $(P_1, P_2)$ between $(x_1, y_1)$ and $(x_2, y_2)$ such that:
\begin{enumerate}
\item the subpaths $P_1[x_1, v]$ and $P_2[w, y_2]$ are internally vertex-disjoint; and
\item the interaction complexity of $P_1[v, y_1]$ and $P_2[x_2, w]$ is at most $k$.
\end{enumerate}

Similarly, we define the polynomials $H_{av,w,k}$, $H_{v,wb,k}$, and $H_{av,wb,k}$, which satisfy the above constraints with the additional requirement that the paths $P_1, P_2$ must contain the edges $(a,v)$, $(w,b)$, and both $(a,v)$ and $(w,b)$, respectively. 

\smallskip

Note that in the definitions of $H_{v,w,k}$, $H_{av,w,k}$, $H_{v,wb,k}$, and $H_{av,wb,k}$ any in-neighbor of $v$ or out-neighbor of $w$ that appears in both $P_1$ and $P_2$ is not counted in the overlap budget $k$ as the overlap constraint in condition two applies exclusively to the number of common vertices between the suffix of $P_1$ starting from $v$ and prefix of $P_2$ up to $w$.

\begin{figure}[!ht]
\centering
\includegraphics[width=0.9\textwidth]{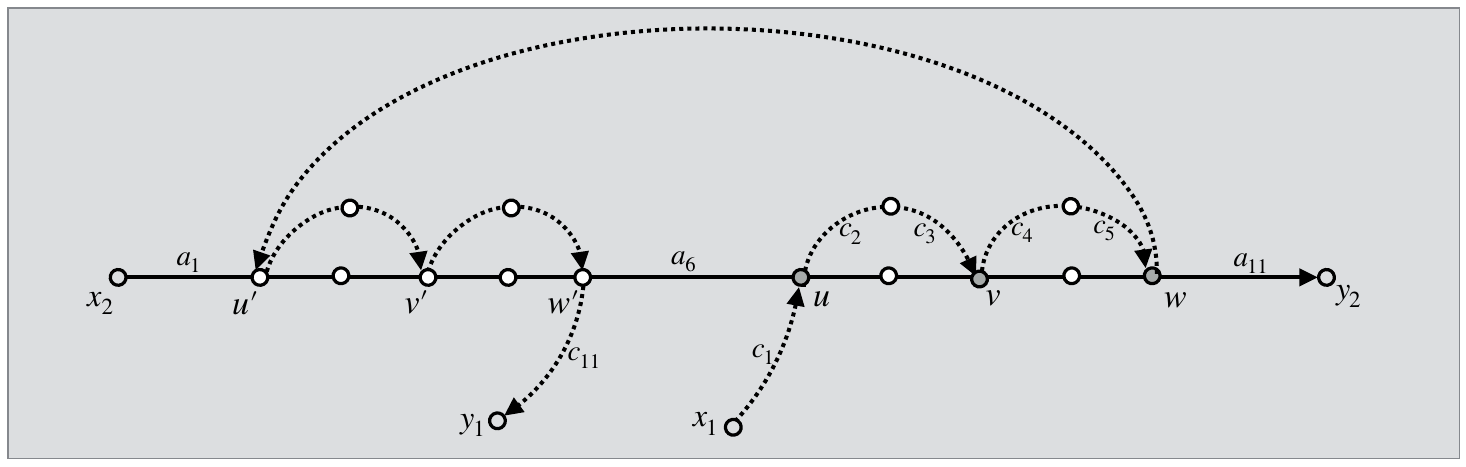}
\caption{
The directed graph above shows multiple shortest paths between $(x_1, y_1)$ (eg. dashed edges $c_1,\ldots,c_{11}$) and $(x_2, y_2)$ (eg. solid edges $a_1,\ldots,a_{11}$);
For any choice of $(x_1,y_1)$ and $(x_2,y_2)$ shortest paths $(P_1, P_2)$, we have $\C(P_1,P_2)=\{(u,w),(u',w')\}$, thereby implying
$F_{\disj,6}(x_1, y_1, x_2, y_2)/(F(u,w)^2\cdot F(u',w')^2) = z_{a_1}z_{a_6}z_{a_{11}}\,z_{c_1}z_{c_6}z_{c_{11}}$.
}
\label{figure:Min-2DSP}
\end{figure}

\subsection{Properties of Helper Polynomials}

We first present a relation between $D_{v,w,k}$ and polynomials $H_{v,w,k}$, $H_{av,w,k}$, $H_{v,wb,k}$, and $H_{av,wb,k}$.

\begin{lemma}
\label{lem:Dvwk}

For any quadruple $\sigma=(x_1,y_1,x_2,y_2)\in V^4$, integer $k\geqslant 1$, and any $v,w\in V$ such that $v$ precedes $w$ on some $(x_1,y_1)$ and $(x_2,y_2)$ shortest paths, the following holds:

$$
D_{v,w,k}(\sigma)
~=~
H_{v,w,k}(\sigma) 
- \hspace{-2mm} \sum_{(a,v) \in E} H_{av,w,k}(\sigma) 
- \hspace{-2mm} \sum_{(w,b) \in E} H_{v,wb,k}(\sigma) 
+ \hspace{-2mm}\sum_{\substack{(a,v), (w,b)\\ \in E}} H_{av,wb,k}(\sigma)
$$
\end{lemma}

\begin{proof}
Let $\A$ be a collection of all pairs of shortest paths $(P_1, P_2)$ between $(x_1, y_1)$ and $(x_2, y_2)$ in which $v$ precedes $w$ such that:
\begin{enumerate}
\item the subpaths $P_1[x_1, v]$ and $P_2[w, y_2]$ are internally vertex-disjoint; 
\item the interaction complexity of $P_1[v, y_1]$ and $P_2[x_2, w]$ is at most $k$; and
\item the in-edges of $v$, as well as the out-edges of $w$, on $P_1$ and $P_2$ are distinct.
\end{enumerate}

Then, by inclusion-exclusion principle, it is easy to observe that the following polynomial 
\begin{equation}
H_{v,w,k}(\sigma) 
- \hspace{-2mm} \sum_{(a,v) \in E} H_{av,w,k}(\sigma) 
- \hspace{-2mm} \sum_{(w,b) \in E} H_{v,wb,k}(\sigma) 
+ \hspace{-2mm}\sum_{\substack{(a,v), (w,b)\\ \in E}} H_{av,wb,k}(\sigma)
\label{eq:Hvwk}
\end{equation}
is enumerating polynomial for $\A$.

\smallskip

Let $\B$ be the set comprising of all pairs of paths $(P_1, P_2)$ contributing to $D_{v,w,k}(\sigma)$, that is, all pairs of shortest paths $(P_1,P_2)$ between $(x_1, y_1)$ and $(x_2, y_2)$ having a interaction complexity of $k$ such that $(v,w)\in \C(P_1,P_2)$ appears before all other concordant pairs on path $P_1$.

\smallskip

In order to prove the lemma it suffices to show that the enumerating polynomial for $\A$ is identical to the enumerating polynomial for $\B$.

\smallskip

We first show that $\B \subseteq \A$. Indeed, if $(P_1, P_2) \in \B$, then $P_1[x_1, v]$ is internally vertex-disjoint from both $P_2[x_2, v]$ and $P_2[v, y_2]$, and thus the in-edges of $v$ in $P_1$ and $P_2$ must be distinct. Similarly, $P_2[w, y_2]$ is internally vertex-disjoint from both $P_1[x_1, w]$ and $P_1[w, y_1]$, and thus the out-edges of $w$ in $P_1$ and $P_2$ must be distinct. Further, interaction complexity of $P_1,P_2$ is bounded by $k$. This shows that all three conditions needed for pairs in $\A$ holds, and thus $\B \subseteq \A$.

\smallskip

Now consider a pair $(P_1,P_2)\in \A$. We study the following two cases.

\begin{description}
\item[{\bf Case 1}] 
No predecessor of $v$ on $P_1$ lies on $P_2$, and
no successor of $w$ on $P_2$ lies on $P_1$.

\smallskip

In this case, $(v,w)$ must be a concordant pair for $P_1,P_2$,  and in fact it will be the first concordant pair in $\C(P_1,P_2)$ that appears on the path $P_1$.  
Further, since interaction complexity for $P_1$ and $P_2$ will be at most $k$, we have $(P_1,P_2)\in \B$. 
Thus, $(P_1,P_2)$ has the same contribution in $D_{v,w,k}(\sigma)$ as in the polynomial in \Cref{eq:Hvwk}.

\medskip

\item[{\bf Case 2}] 
A predecessor of $v$ on $P_1$ lies on $P_2$, and/or
a successor of $w$ on $P_2$ lies on $P_1$.

\smallskip

Without loss of generality, let us suppose that a predecessor of $v$ on $P_1$ also lies on $P_2$.  
Let $r$ denote the immediate predecessor of $v$ on $P_1$ lying on $P_2$.
Since $P_1[x_1,v]$ is internally disjoint from $P_2[w,y_2]$, it follows that $r \in P_2[x_2,w]$. This along with \Cref{lem:order-intersect}, implies that $r$ must be the immediate predecessor of $v$ on $P_2$.  
Hence, $(v,w)$ cannot be a concordant pair with respect to $(P_1,P_2)$, and therefore $(P_1,P_2) \in \A\setminus \B$.  

Moreover, $(r,v)$ must be a twin-crossing pair for $(P_1,P_2)$, since the in-edges of $v$ on $P_1$ and $P_2$ are distinct, which implies $P_1[r,v] \neq P_2[r,v]$.  
By the definition of $\A$ and $\B$, the pair
$(Q_1,Q_2) \;=\; \phi_{r,v}(P_1,P_2)$,
obtained by performing a swap at $(r,v)$, also lies in $\A\setminus \B$.  Further, it follows that $r$ must be the immediate predecessor of $v$ on $Q_1,Q_2$ and $(P_1, P_2) = \phi_{r,v}(Q_1, Q_2)$.



So, the set $\A \setminus \B$ can be partitioned into pairs 
\[
\{(P_1, P_2), (Q_1, Q_2)\},
\]
where $(Q_1, Q_2) = \phi_{r,v}(P_1, P_2)$ and $r$ is the immediate predecessor of $v$ on $P_1,P_2$ (as well as in $Q_1,Q_2$).  
Since $\phi_{r,v}$ preserves the monomial contribution, we see that each such pair contributes two identical monomials to the enumerating polynomial of $\A \setminus \B$. Therefore, the enumerating polynomial of $\A \setminus \B$ is zero in a field of characteristic two.
\end{description}

\smallskip

This proves the desired claim that the enumerating polynomial of $\A$ is identical to the enumerating polynomial of $\B$.
\end{proof}

Next we provide construction of other helper polynomials. 

\begin{lemma}
\label{lem:min2dsp-id1}
For any integer $k$, and any $v,w\in V$ such that 
$v$ precedes $w$ on some $(x_1,y_1)$ and $(x_2,y_2)$ shortest paths, 
the following holds:
\begin{align*}
H_{v,w,k}(x_1,y_1,x_2,y_2) 
= 
F_{\disj}(x_1,v,w,y_2)
\cdot 
F(v,w)^2
\cdot 
F_{\disj,k-\delta(v,w)}(w,y_1,x_2,v) 
\end{align*}
\end{lemma}

\begin{proof}
Consider a pair of shortest paths $( P_1, P_2 )$ contributing to $H_{v,w,k}(x_1,y_1,x_2,y_2)$. For $i=1,2$, decompose $P_i = Q_i \cdot L_i \cdot R_i$, where $Q_i = P_i[x_i, v]$, $L_i=P_i[v,w]$, and $R_i = P_i[w, y_i]$ are shortest paths.
By definition of polynomial $H_{v,w,k}$, the interaction complexity of $Q_2,R_1$ must be bounded by $k-\delta(v,w)$. The contributions of the paths $Q_2$ and $R_1$ is captured by $F_{\disj,k-\delta(v,w)}(w,y_1,x_2,v)$. 
Meanwhile, the contribution of all the internally disjoint path pairs $Q_1, R_2$ is given by the polynomial $F_{\disj}(x_1, v, w, y_2)$. 
Finally, observe that the contribution of $L_1,L_2$ is given by $F(v,w)^2$ since if $L_1\neq L_2$, then the pair $(P_1',P_2')=\phi_{v,w}(P_1,P_2)$ will have the same monomial contribution as that of $(P_1,P_2)$, which will cancel modulo $2$.
Multiplying these  components yields the desired expression.
\end{proof}

Similar to the above lemma, we can establish the following.

\smallskip

\begin{lemma}
\label{lem:min2dsp-id2}
For any integer $k$, and any $(a,v)\in E$, $w\in V$ such that $(a,v)$ precedes $w$ on some $(x_1,y_1)$ and $(x_2,y_2)$ shortest paths, the following holds:
\begin{align*}
H_{av,w,k}(x_1,y_1,x_2,y_2) = 
z_{av}^2 
\cdot 
F_{\disj}(x_1,a,w,y_2)
\cdot 
F(v,w)^2
\cdot
F_{\disj,k-\delta(v,w)}(w,y_1,x_2,a)
\end{align*}
\end{lemma}

\smallskip

\begin{lemma}
\label{lem:min2dsp-id3}
For any integer $k$, and any $v\in V$, $(w,b)\in E$ such that $v$ precedes $(w,b)$ on some $(x_1,y_1)$ and $(x_2,y_2)$ shortest paths, the following holds:
\begin{align*}
H_{v,wb,k}(x_1,y_1,x_2,y_2) = 
z_{wb}^2 
\cdot 
F_{\disj}(x_1,v,b,y_2)
\cdot 
F(v,w)^2
\cdot
F_{\disj,k-\delta(v,w)}(b,y_1,x_2,v)
\end{align*}
\end{lemma}

\smallskip

\begin{lemma}
\label{lem:min2dsp-id4}
For any integer $k$, and any $(a,v),(w,b)\in E$ such that $(a,v)$ precedes $(w,b)$ on some $(x_1,y_1)$ and $(x_2,y_2)$ shortest paths, the following holds:
\begin{align*}
H_{av,wb,k}(x_1,y_1,x_2,y_2) = 
z_{av}^2\cdot z_{wb}^2 
\cdot 
F_{\disj}(x_1,a,b,y_2)
\cdot 
F(v,w)^2
\cdot
F_{\disj,k-\delta(v,w)}(b,y_1,x_2,a) 
\end{align*}
\end{lemma}

\subsection{An $O(m^2 n^3)$ time algorithm for Min-2-DSP}
In this subsection, we present an algorithm for solving the Min-2-DSP problem in weighted directed graphs that takes $O(m^2n^3)$ time. Let $(s_1, t_1)$ and $(s_2, t_2)$ be the input terminal pairs. By~Schwartz-Zippel lemma (\Cref{lem:SZ}), in order to check if $\barFdk(s_1, t_1, s_2, t_2)$ is non-zero it suffices to evaluate this polynomial  on a random assignment of the variables $\{z_e\}_{e \in E}$ drawn from a sufficiently large finite field $\mathbb{F}$. The algorithm outputs the smallest integer $k$ for which $\barFdk(s_1, t_1, s_2, t_2)$ evaluates to non-zero at a random assignment of edge variables.

We begin by assigning random values $\barz_e$ from $\mathbb{F}$ to each of the edge variables $z_e$ respectively. As before, for a polynomial $R$ involving the variables $z_e, e \in E$, we shall use ${\bar R}$ to denote the value in $\mathbb{F}$ obtained by evaluating $R$ at $\barz_e, e \in E$.  
Using \Cref{lem:Fxy}, we compute $\barF(x, y)$ for all pairs of vertices in $O(mn \log n)$ time. 

\smallskip

We next evaluate the polynomial $\barFd(x_1,x_2,y_1,y_2)$, for each quadruple $(x_1,x_2,y_1,y_2)\in V^4$, in $O(mn^4)$ total time. This is accomplished by leveraging a dynamic programming approach similar to Algorithm~\ref{algo:2-DSP-directed} that processes the quadruples in  order of the sum total distance between the endpoints (i.e. $\dist(x_1,y_1)+\dist(x_2,y_2)$). 

At each step, we use recurrence identities established in \Cref{lem:Hv,lem:Hav,lem:Dv,lem:F-disj}. For a fixed $(x_1, y_1, x_2, y_2)$, we compute $\barD_v(x_1, y_1, x_2, y_2)$ for each valid $v$ in $O(\indeg(v))$ time, and combine them to evaluate $\barF_{\disj}(x_1, y_1, x_2, y_2)$ in $O(m)$ time. This yields an overall runtime of $O(mn^4)$. See Algorithm~\ref{algo:all-pairs-2-DSP-directed} for details.

\begin{algorithm}[!ht]
\setstretch{1.25}

\ForEach{$(x_1, y_1, x_2, y_2) \in V^4$ in increasing order of $\dist(x_1, y_1) + \dist(x_2, y_2)$}{
  \ForEach{$v \in V \setminus (\{x_1,y_1\}\cap\{x_2,y_2\})$}{
    \uIf{$v\in V(x_1, y_1) \cap V(x_2, y_2)$}{Compute $\barD_v(x_1,y_1,x_2,y_2)$ as
    $$\hspace{-2.5mm}
    \barF(x_2,v) \cdot \barFd(x_1, v, v, y_2) \cdot \barF(v, y_1)~~~
    - 
    \hspace{-2mm}
    \sum_{\substack{(a,v) ~\in \\ E(x_1,y_1),\\  E(x_2,y_2)}}
    \hspace{-1mm}
    z_{av}^2 \cdot F(x_2,a) \cdot \barFd(x_1,a, v, y_2) \cdot \barF(v, y_1);
    $$
    \vspace{-8mm}
    }
    \Else{Set $\barD_v(x_1,y_1,x_2,y_2)=0;$}
  }
  Compute
  
  $$
  \barFd(x_1, y_1, x_2, y_2) = \barF(x_1, y_1) \cdot \barF(x_2, y_2)
  -
  \sum_{\substack{v \in V \setminus (\{x_1,y_1\}\cap\{x_2,y_2\})}} \barD_v(x_1, y_1, x_2, y_2);
  $$
  \vspace{-2mm}
}
\caption{\textsc{All-Pairs-2-DSP-Directed}($G$)}
\label{algo:all-pairs-2-DSP-directed}
\end{algorithm}

\smallskip

In order to compute $\barFdk(s_1,t_1, s_2, t_2),$ we would need to evaluate $\barFdk(x_1,y_1, x_2, y_2)$ for several values of tuples $(x_1, y_1, x_2, y_2)$ and integers $k$. We compute a list $\cL$ containing tuples of the form $(v, t_1, s_2, w)$ for all $v, w \in V$. For each such tuple in $\cL$, we evaluate $\barFdk(x_1, y_1, x_2, y_2)$ using dynamic programming. As before, these computations are performed in the increasing order of the sum total distance between the endpoints, i.e. $\dist(x_1, y_1) + \dist(x_2, y_2)$, which ensures that all necessary values for the recurrence are available when needed.
For base cases, when either $x_1 = y_1$ or $x_2 = y_2$, , the values simplify as follows:
\begin{enumerate}
\item $\barFdk(x_1,x_1, x_2, y_2) = \barFdk(x_2, y_2)$, and 
\item $\barFd(x_1,y_1, x_2, x_2) = \barF(x_1, y_1)$. 
\end{enumerate}

\bigskip

To evaluate $\barFdk(x_1, y_1, x_2, y_2)$, we use the identity (see~\Cref{lem:min2dsp-id0}):
$$
\barFdk(x_1,y_1,x_2,y_2)
=
\sum_{\substack{v,w \in V
}}
\barD_{v,w,k}(x_1, y_1, x_2, y_2)
$$

By \Cref{lem:Dvwk}, for any $\sigma=(x_1,y_1,x_2,y_2)$ and any $v,w\in V$ satisfying $dist(x_i,y_i)=dist(x_i,v)+dist(v,w)+dist(w,y_i)$, for $i=1,2$, $\barD_{v,w,k}$ can be expressed as

$$
\barD_{v,w,k}(\sigma)
~=~
\barH_{v,w,k}(\sigma) 
- \hspace{-2mm} \sum_{(a,v) \in E} \barH_{av,w,k}(\sigma) 
- \hspace{-2mm} \sum_{(w,b) \in E} \barH_{v,wb,k}(\sigma) 
+ \hspace{-2mm}\sum_{\substack{(a,v), (w,b)\\ \in E}} \barH_{av,wb,k}(\sigma).
$$

\begin{algorithm}[!ht]
\setstretch{1.25}
Assign each edge $(u,v) \in E$ an independent and uniformly random value $\barz_{uv}$ from $\mathbb{F}$\;
Use Lemma~\ref{lem:Fxy} to compute $\barF(x,y)$ for all $x, y \in V$\;
Initialize a list $\cL$ with all quadruples of the form $(v,t_1,s_2,w)$ for all $v,w \in V$\; 
\ForEach{$(x_1, y_1, x_2, y_2) \in \cL$ in increasing order of $\dist(x_1, y_1) + \dist(x_2, y_2)$}{
  \lIf{$x_1=y_1$ \textup{and} $x_2=y_2$}{Set $\barFdk(x_1,y_1, x_2, y_2)=1$ for all $k$, and \textbf{continue} for-loop}
  %
  \vspace{1mm}
  \ForEach{$1\leqslant k\leqslant n$ {\bf and} $v,w \in V \setminus (\{x_1,y_1\}\cap\{x_2,y_2\})$}{
    \uIf{$(\dist(x_i,v)+\dist(v,w)+\dist(w,y_i) = \dist(x_i, y_i),$ for $i=1,2)$}{Compute $\barD_{v,w,k}(x_1,y_1,x_2,y_2)$ as $   \barF(v,w)^2$ times
    \begin{align*}        
    & \quad~\barFd(x_1,v,w,y_2)\cdot \bar F_{\disj,k-\delta(v,w)}(w,y_1,x_2,v) 
    \\[2mm] 
    -&~ \sum_{\substack{(a,v) ~\in \\ E(x_1,y_1),\\  E(x_2,y_2)}} 
    z_{av}^2 \cdot 
    \barFd(x_1,a,w,y_2) \cdot \bar F_{\disj,k-\delta(v,w)}(w,y_1,x_2,a)
    \\ 
    -&~ \sum_{\substack{(w,b) ~\in \\ E(x_1,y_1),\\  E(x_2,y_2)}}
    z_{wb}^2 \cdot 
    \barFd(x_1,v,b,y_2) \cdot \bar F_{\disj,k-\delta(v,w)}(b,y_1,x_2,v)
    \\
    +& \sum_{\substack{(a,v),(w,b) \\ ~\in E(x_1,y_1),\\ ~~E(x_2,y_2)}} (z_{av} z_{wb})^2 \cdot 
    \barFd(x_1,a,b,y_2) \cdot \bar F_{\disj,k-\delta(v,w)}(b,y_1,x_2,a); 
    \end{align*}

    \vspace{-4mm}
    }
    \Else{Set $\barD_{v,w,k}(x_1,y_1,x_2,y_2)=0;$}
  }
   {\bf foreach~}{$1\leqslant k\leqslant n$}{\bf ~do~}{\ Compute
  $$
  \barFdk(x_1, y_1, x_2, y_2) = 
  \sum_{\substack{v,w \in V \setminus (\{x_1,y_1\}\cap\{x_2,y_2\})}} \barD_{v,w,k}(x_1, y_1, x_2, y_2);
  $$}
}
{\bf Return}
the smallest integer $k\geqslant 1$ for which $\barFdk(s_1, t_1, s_2, t_2) \neq 0$\;
\caption{\textsc{Min-2-DSP-Directed}($G, s_1, t_1, s_2, t_2$)}
\label{algo:min-2-DSP-directed}
\end{algorithm}

Now, by \Cref{lem:min2dsp-id1,lem:min2dsp-id2,lem:min2dsp-id3,lem:min2dsp-id4}, we have
\begin{align*}
    \barD_{v,w,k}(x_1,y_1,x_2,y_2)~=
    & ~~\barFd(x_1,v,w,y_2)\cdot 
    \barF(v,w)^2\cdot 
    \bar F_{\disj,k-\delta(v,w)}(w,y_1,x_2,v) 
    \\[2mm] 
    &~- \sum_{\substack{(a,v) ~\in \\ E(x_1,y_1),\\  E(x_2,y_2)}} 
    z_{av}^2 \cdot     
    \barFd(x_1,a,w,y_2) \cdot 
    \barF(v,w)^2\cdot     
    \bar F_{\disj,k-\delta(v,w)}(w,y_1,x_2,a)
    \\ 
    &~- \sum_{\substack{(w,b) ~\in \\ E(x_1,y_1),\\  E(x_2,y_2)}}
    z_{wb}^2 \cdot 
    \barFd(x_1,v,b,y_2) \cdot
    \barF(v,w)^2\cdot 
    \bar F_{\disj,k-\delta(v,w)}(b,y_1,x_2,v)
    \\
    &+ \sum_{\substack{(a,v),(w,b) \\ ~\in E(x_1,y_1),\\ ~~E(x_2,y_2)}} (z_{av} z_{wb})^2 \cdot 
    \barFd(x_1,a,b,y_2) \cdot
    \barF(v,w)^2\cdot 
    \bar F_{\disj,k-\delta(v,w)}(b,y_1,x_2,a). 
\end{align*}

Note that $\dist(w,y_1) + \dist(x_2,v) < \dist(x_1, y_1) + \dist(x_2, y_2)$ and hence, the r.h.s. has been computed already. 
Finally, we return the smallest integer $k$ for which $\barFdk(s_1, t_1, s_2, t_2)$ evaluates to non-zero.

\paragraph{Correctness and Running Time} 
The correctness of the procedure follows from the recursive expressions in \Cref{lem:min2dsp-id0,lem:Dvwk,lem:min2dsp-id1,lem:min2dsp-id2,lem:min2dsp-id3,lem:min2dsp-id4}. 

We now analyze the running time.
From the identities above, observe that each polynomial $\barD_{v,w,k}(x_1, y_1, x_2, y_2)$ can be evaluated in $O(\indeg(v)\cdot \outdeg(w))$ time for any fixed $v, w$, and integer $k$. Consequently, we can compute $\barFdk(x_1, y_1, x_2, y_2)$ in $O(m^2)$ time for any fixed quadruple $(x_1, y_1, x_2, y_2)$ and any fixed $k$. Since we evaluate these polynomials for $O(n^2)$ such quadruples in the set $\cL$ and for $n$ distinct values of $k$, the total evaluation time is bounded by $O(m^2 n^3)$.

The algorithm has an initial preprocessing phase, where we first solve All-Pairs 2-DSP in $O(mn^4)$ time, and then compute the number of distance-critical vertices $\delta(x, y)$ for all pairs $(x, y)$ in $O(mn)$ time. The latter uses the dominator-tree algorithm of Lengauer and Tarjan~\cite{LT79,GeorgiadisT05}, which identifies all cut vertices for pairs in $\{s\} \times V$ in $O(m+n)$ time per source.

Altogether, this yields a randomized algorithm for Min-2-DSP with a total running time of $O(m^2 n^3)$ and high success probability.

\smallskip

The details of our procedure for reporting the actual $(s_1,t_1)$ and $(s_2,t_2)$ shortest paths with minimum intersection are deferred to the full version, as it is closely related to the procedure outlined in \Cref{sec:reportingpaths2DSP}. We conclude with the following theorem.

\mintwoDSPdirected*

\section{Min-2-DSP in DAGs}
\label{sec:Min-2DSP-DAGs}

We now present an $O(m+n)$-time algorithm for the Min-2-DSP problem in DAGs. 
For simplicity, we assume that the vertices $s_1, s_2, t_1, t_2$ are distinct, and that no vertex-disjoint shortest paths exist for the pairs $(s_1, t_1)$ and $(s_2, t_2)$. This assumption can be verified in $O(m + n)$ time, as shown by the work of \cite{AkmalVW24}.

Somewhat surprisingly, we show that if there exist no vertex-disjoint shortest paths for the pairs $(s_1, t_1)$ and $(s_2, t_2)$, then one can report the explicit pair of $(s_1, t_1)$ and $(s_2, t_2)$ shortest paths which intersect at minimum number of vertices in linear time deterministically (note that there is no known $o(mn)$ time randomized / deterministic algorithm for explicitly reporting the shortest paths if they are vertex-disjoint). 

\smallskip

\smallskip

Consider the subgraph $\G = (V, E_\cap)$ of the input DAG $G$, where $E_\cap=E(s_1, t_1)\cap E(s_2, t_2)$ is the set of edges that lie in some $(s_1,t_1)$-shortest-path as well as some $(s_2,t_2)$-shortest-path. 
Let $\G^{un}$ denote the undirected realization of $\G$, and let $\B$ represent the set of connected components of this undirected graph.
Below, we outline some important properties of graph $\G$ and the family $\B$.

\begin{lemma}
    Let $(P_1,P_2)\in \Pi(s_1,t_1)\times \Pi(s_2,t_2)$ be a pair of paths, and $x,y\in V$ be such that $x$ precedes $y$ on both the paths.
    Then, the edges of the subpaths $P_1[x,y],P_2[x,y]$ are contained in $E_{\cap}$. 
    \hspace{-1mm}
    
\label{lemma:global-to-local-helper}
\end{lemma}
\begin{proof}
    If the segments $P_1[x,y], P_2[x,y]$ are identical, then the claim holds trivially. So, we assume that $P_1[x,y], P_2[x,y]$ are not identical.
    Let $(Q_1, Q_2) = \phi_{x,y}(P_1,P_2)$ be the paths obtained by swapping the segments  $P_1[x,y]$ and $P_2[x,y]$ in $P_1,P_2$.
    \Cref{cl:crossing} shows that $(Q_1, Q_2)\in \Pi(s_1,t_1)\times \Pi(s_2,t_2)$. Hence, the edges of subpaths $P_1[x,y],P_2[x,y]$ are contained in $(s_1,t_1)$ as well as $(s_2,t_2)$ shortest paths, thereby proving the lemma.
\end{proof}

\Cref{lemma:global-to-local-helper} helps us establish the following result, which shows that any two $(s_1, t_1)$ and $(s_2, t_2)$ shortest paths can intersect only in the vertices of a single component in $\B$.

\begin{lemma}
   Let $(P_1,P_2)\in \Pi(s_1,t_1)\times \Pi(s_2,t_2)$ be a pair of intersecting paths. Then, there exists a component $B \in \B$ such that $V(P_1) \cap V(P_2) \subseteq B$.  
   
   \label{lemma:global-to-local-paths}
\end{lemma}
\begin{proof}
    Let $P_1, P_2$ be a pair of intersecting $(s_1, t_1), (s_2, t_2)$ shortest paths. 
    If $|V(P_1) \cap V(P_2)| = 1$, then the claim holds trivially. 
    So, consider a pair of vertices $x,y$ such that $x$ precedes $y$ in $P_1,P_2$. By \Cref{lemma:global-to-local-helper}, we get that $E(P_1[x,y]) \subseteq E_\cap$. 
    Therefore, $x,y$ must be connected in $\G^{un}$.
\end{proof}

The above lemma suggests a natural approach to solving the Min-2-DSP problem in DAGs. For each $B \in \B$, we search for a pair of $(s_1, t_1)$ and $(s_2, t_2)$ shortest paths that intersect in the minimum possible number of nodes in $B$. To achieve this, we define the graph $H_B$ for every component $B\in \B$, in the following manner.

\begin{enumerate}
\item 
Initialise $H_B = \G[B]$, that is, $H_B$ is the induced subgraph of $\G$ on the vertices of $B$.

\item Next, introduce dummy sources $s_{B,1}, s_{B,2}$ and dummy sinks $t_{B,1}, t_{B,2}$.


\item For each edge $(x,y) \in E$ {\em incoming} to $B$, i.e., $x \notin B, y \in B$, if there is an $i\in\{1,2\}$ such that $e \in E(s_i, t_i)$, then add the edge $(s_{B,i}, y)$ to $H_B$. Similarly, we add outgoing edges from appropriate nodes in $B$ to terminals $t_{B_1}$ and $t_{B_2}$.


\item Finally, introduce two additional vertices $s_B,t_B$, and add directed edges from $s_B$ to $s_{B,1}, s_{B,2}$, as well as directed edges from $t_{B,1}, t_{B,2}$ to $t_B$.
\end{enumerate}

\begin{figure}[!ht]
\centering
\includegraphics[height=70mm, page=1]{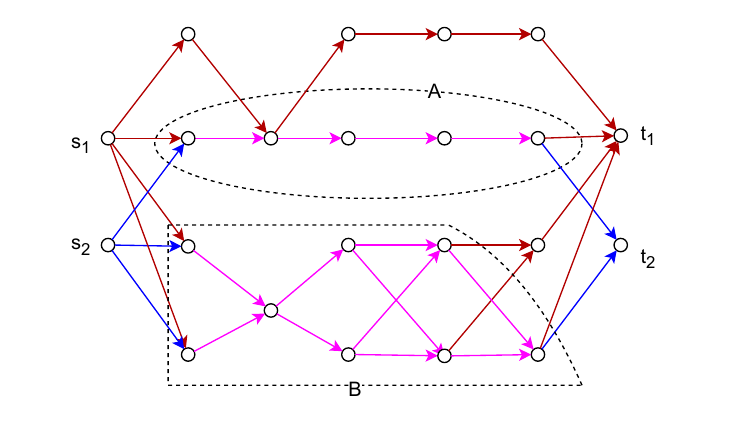}
\quad~
\includegraphics[height=32mm, page=1]{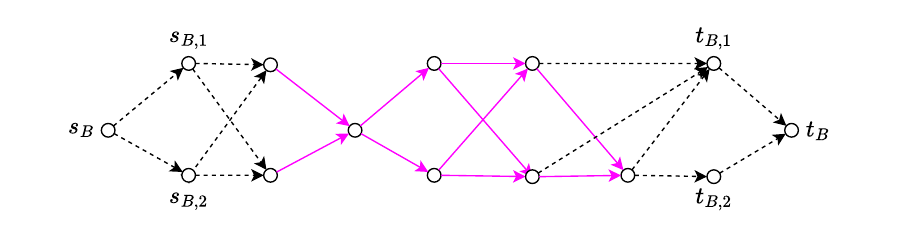}
\caption{
The first figure illustrates $\B$. The pink edges belong to both $(s_1, t_1)$ as well as $(s_2, t_2)$ shortest paths, assuming unit edge weights. The graph $\G^{un}$ has two non-singleton connected components, $A$ and $B$. The second figure illustrates the graph $H_{B}$ for component $B$. 
}
\label{fig:DAG-components-eg}
\end{figure}

Observe that the graph $H_B$ is a DAG.

\begin{lemma}
Let $P$ be an $(s_i, t_i)$ shortest path in $G$, for some $i\in\{1,2\}$, and let $x,y\in B\in \B$ be such that $x$ precedes $y$ in $P$. Then the subpath $P[x,y]$ lies in $H_B$.
\end{lemma}

\begin{proof}
We prove the claim for $i=1$; the proof for the case $i=2$ follows similarly.
Consider an $(s_1,t_1)$ shortest path $P$ in $G$, passing through a set $B\in \B$. Let $x,y\in B$ be such that $x$ precedes $y$ in $P$. 
By definition of $\W$, it follows that 
\begin{equation}
\dist(s_1,y)-\dist(s_1,x)
=
\dist(s_2,y)-\dist(s_2,x).
\label{eq:dist-diff}
\end{equation}

Now, consider any $(s_2,t_2)$ shortest paths $Q_x, Q_y$ passing through $x$ and $y$ respectively. We show that the path $Q = Q_x[s_2, x] \cdot P[x,y] \cdot Q_y[y, t_2]$ is also an $(s_2, t_2)$ shortest path. This can be shown by finding weight of the path $Q$:
\begin{align*}
    wt(Q) &= \dist(s_2, x) + \dist(y, t_2) + \dist(x,y) \\
    &= \dist(s_2, x) - \dist(s_2, y) + \dist(s_2, t_2) + \dist(x,y)   \\
    &= \dist(s_{\pmb{1}}, x) - \dist(s_{\pmb{1}}, y) + \dist(x,y) + \dist(s_2, t_2)  & \text{(From Eq.~\ref{eq:dist-diff})}\\
    &= \dist(s_2, t_2)
\end{align*}

Since $Q$ is an $(s_2,t_2)$ shortest path, the edges of $P[x,y]$ must lie in $E_2$ and hence, also in $E_\cap$. This proves that $P[x,y]$ lies in $H_B$.
\end{proof}

\begin{definition}[Projection]
Let $P$ be any $(s_i, t_i)$ shortest path in $G$, for some $i\in\{1,2\}$, passing through a set $B\in \B$. Suppose
$x$ and $y$ are the first and last vertices of $P$ lying in $B$, respectively. We define the projection of the path $P$ onto graph $H_B$ as the path $$P_B = (s_{B_,i}, x) \cdot P[x,y] \cdot (y, t_{B,i}).$$ 
\end{definition}

\medskip

We now present a key lemma that highlights the significance of the projections defined above.

\smallskip

\begin{lemma}
    There exist pair of paths $(P_1,P_2)\in \Pi(s_1,t_1)\times \Pi(s_2,t_2)$ intersecting at exactly $t \geq 1$ vertices if and only if there exists a component $B\in \B$ such that there exist $(s_{B,1}, t_{B,1}), (s_{B,2}, t_{B,2})$ paths $Q_1, Q_2$ in $H_B$ with intersecting at exactly $t$ vertices. 
    Furthermore, $V(P_1) \cap V(P_2) = V(Q_1) \cap V(Q_2)$, and $Q_1,Q_2$ are precisely the paths obtained by projecting $P_1, P_2$ onto $H_B$. 
    \label{lemma:reduction-to-local}
\end{lemma}

The above lemma helps us to focus on {\em paths} in graph $H_B$, for $B\in \B$, instead of the {\em shortest paths} in $G$.

\paragraph{Reduction from Minimum Disjoint Shortest Paths to Minimum Disjoint Paths.}
We now describe how the problem of finding two shortest paths in $G$ with minimum intersection reduces to the Minimum \emph{Disjoint Paths} (Min-2-DP) problem, assuming no two vertex-disjoint $(s_1, t_1)$ and $(s_2, t_2)$ shortest paths exist in~$G$. See Algorithm~\ref{algo:min-2-DSP-DAG}.

For each block $B \in \B$, we construct the auxiliary graph $H_B$ and compute the minimum number of intersections between any $(s_{B,1}, t_{B,1})$ and $(s_{B,2}, t_{B,2})$ paths (the intersections will be at least one in count due to our assumption on $G$).
For each vertex $v$ in $H_B$, we set $\alpha(v)$ to the total number of cut vertices for the pairs $(s_B, v)$ and $(v, t_B)$, excluding the endpoints $s_B$ and $t_B$. This value over all $v\in V(H_B)$ can be computed using the dominator trees of $s_B$ and $t_B$ in $O(|V(H_B)|+|E(H_B)|)$ time~\cite{LT79,GeorgiadisT05}. The quantity $\alpha(v)$ represents the minimum number of intersections between any pair of $(s_{B,1}, t_{B,1})$ and $(s_{B,2}, t_{B,2})$ paths, that pass through vertex $v$.
To solve the 2-DP problem for the graph $H_B$, we take the minimum $\alpha(v)$ over all vertices $v \in H_B$. Finally, by selecting the smallest intersection value across all auxiliary graphs $\{H_B \mid B \in \B\}$, we obtain the optimal solution to the \textsc{Min-2-DSP} problem in $G$, due to \Cref{lemma:reduction-to-local}.

\begin{algorithm}[!ht]
\setstretch{1.25}
\lIf{$\textsc{2-DSP}(s_1,t_1,s_2,t_2)$}{Return $0$}
Compute $E_\cap = E(s_1, t_1)\cap E(s_2, t_2)$ and define $\G=(V,E_\cap)$\;
Let $\B$ be a partitioning of $V$ induced by the connected components of $\G^{un}$\;
\ForEach{$B\in \B$}{
    Initialize $H_B$ to $\G[B]$, and 
    add 6 additional vertices $s_B, s_{B,1}, s_{B,2},t_B,t_{B,1}, t_{B,2}$ to $H_B$\;
    \ForEach{$i \in \{1,2\}$}{
        \ForEach{$(x,y) \in E(s_i, t_i)$ incoming to set $B$}{
            Add the edge $(s_{B,i}, y)$ to $H_B$\; 
        }

        \ForEach{$(x,y) \in E(s_i, t_i)$ outgoing from set $B$}{
            Add the edge $(x, t_{B,i})$ to $H_B$\; 
        }

        Add the edges $(s_B, s_{B,i}), (t_{B,i}, t_B)$ to $H_B$\;
    }

    Compute $T_{out}=\domtree(s_B,H_B)$\;
    Compute $T_{in}=\domtree(t_B,(H_B)^{rev})$\;
    \lForEach{$v \in B$}{
        Set $\alpha(v) = \depth(v,T_{in}) + \depth(v,T_{out})-1$
        }
}

Return $\min_{v\in V\setminus\{s_1,t_1,s_2,t_2\}} \alpha(v)$
\caption{\textsc{Min-2-DSP-DAG}($G, s_1, t_1, s_2, t_2$)}
\label{algo:min-2-DSP-DAG}
\end{algorithm}

\footnotetext{Depth of nodes not lying in the dominator trees $T_{out},T_{in}$ can be taken as infinity.}


The running time of Algorithm~\ref{algo:min-2-DSP-DAG} is $\sum_{B\in \B}O(|V(H_B)|+|E(H_B)|)$ which is just $O(m + n)$.


\subsection{Reporting the Paths} 
For any $v \in V(s_1, t_1) \cap V(s_2, t_2)$, suppose we define \textsc{subproblem}$(v)$ as the following: Find the $(s_1, t_1), (s_2, t_2)$ shortest paths $(P_1, P_2)$ respectively, such that $v \in V(P_1) \cap V(P_2)$ and $|V(P_1) \cap V(P_2)|$ is minimised. 

Now, we show how to solve \textsc{subproblem}$(v)$ in $O(m+n)$ time. For any vertex $v \in V$, define two subgraphs of $G$ as below:
\begin{enumerate}
    \item $G^{in}_v$ contains an edge $(x,y)$ if and only if $(x,y) \in E(x,v)$, i.e., $(x,y)$ lies on some shortest path from $x$ to $v$.  
    \item $G^{out}_v$ contains an edge $(x,y)$ if and only if $(x,y) \in E(v,y)$, i.e., $(x,y)$ lies on some shortest path from $v$ to $y$.  
\end{enumerate}

The computation of these subgraphs for each $v$ takes $O(m+n)$ time, since single-source shortest path distances can be computed in $O(m+n)$ time in DAGs. The following lemma follows directly from definitions above.  

\begin{lemma}
Any path from $x$ to $v$ in $G^{in}_v$ is a shortest $(x,v)$ path in $G$, and  vice-versa. Similarly, any path from $v$ to $y$ in $G^{out}_v$ is a shortest $(v,y)$ path in $G$, and vice-versa.
\end{lemma}

\paragraph{Solving \textsc{subproblem}$(v)$.}  
To solve \textsc{subproblem}$(v)$, we must address two tasks:  
\begin{enumerate}
    \item Find $(v, t_1), (v, t_2)$ shortest paths $Q_1, Q_2$ in $G^{out}_v$ that intersect in the minimum possible number of vertices.
    \item Find $(s_1, v), (s_2, v)$ shortest paths $Q_1', Q_2'$ in $G^{in}_v$ that intersect in the minimum possible number of vertices.      
\end{enumerate}

It can be seen that the paths $Q_1' \cdot Q_1$ and  $Q_2' \cdot Q_2$ will be the desired answer to \textsc{subproblem}$(v)$.

We now discuss how to find the paths $Q_1, Q_2$ in $O(m+n)$ time (the paths $Q_1', Q_2'$ can be found by similar procedure).
To compute $Q_1$ and $Q_2$, we add a dummy sink $t^*$ and include edges $(t_1,t^*)$ and $(t_2,t^*)$. The task reduces to finding two $v$ to $t^*$ paths that intersect at the smallest possible number of vertices. This can be solved in $O(m+n)$ time using the dominator-tree algorithm of Lengauer and Tarjan~\cite{LT79,GeorgiadisT05}, which identifies all cut vertices between two vertex-pairs in $O(m+n)$ time. Given the cut vertices, we can construct $Q_1$ and $Q_2$ by computing two vertex-disjoint paths between each pair of consecutive cut vertices.
Hence, \textsc{subproblem}$(v)$ can be solved in $O(m+n)$ time.

After identifying the vertex $v^* = \arg \min_{v \in V} \alpha(v)$, one can report the explicit paths intersecting at minimum number of vertices by solving \textsc{subproblem}($v^*$) which takes $O(m+n)$ time as discussed previously. This results in the following theorem.

\mintwoDSPDAGs*

\section{Min-2-DSP in Undirected Graphs}
\label{sec:Min-2DSP-undirected}

We present in this section an $O(m + n)$-time algorithm for the Minimum-2-Disjoint Shortest Paths (Min-2-DSP) problem in weighted (strictly positive weight) undirected graphs.  Throughout this section, all instances of single-source shortest path computation are handled via Thorup's algorithm~\cite{Thorup97} for undirected graphs, so each such subroutine runs in $O(m+n)$ time.
For simplicity, we assume that the vertices $s_1, s_2, t_1, t_2$ are all distinct, and that there do not exist two vertex-disjoint shortest paths for the pairs $(s_1, t_1)$ and $(s_2, t_2)$. This assumption can be verified in $O(m + n)$ time, as shown in~\cite{AkmalVW24}. 


We begin by defining the concepts of agreeing and disagreeing paths, similar to that in~\cite{AkmalVW24}.

\begin{definition}
Let $P_1$ and $P_2$ be shortest paths from $s_1$ to $t_1$ and from $s_2$ to $t_2$, respectively. Let $W$ be the set of vertices common to both $P_1$ and $P_2$.
\begin{itemize}
    \item If $|W| \geq 2$, we say that $(P_1, P_2)$ is an {agreeing pair} if the vertices of $W$ appear in the same order on both $P_1$ and $P_2$, and a {disagreeing pair} if the order of $W$ is reversed on $P_1$ and $P_2$.
    \item If $|W| = 1$, we say that $P_1$ and $P_2$ have a \emph{singular intersection}.
\end{itemize}
\end{definition}

\subsection{Handling Agreeing/Disagreeing Pairs}
In this subsection, we note that the Min-2-DSP idea of DAGs extends to undirected graphs with positive edge weights, as long as minimum number of intersections between any pair of shortest paths is at least two. We describe the idea in brief:


Suppose $E_\cap = E_1 \cap E_2$. Now, we partition $E_\cap$ into $E_\cap^+$ and $E_\cap^-$ as following:~\footnote{This is the same decomposition as used by Lochet~\cite{Lochet21}.} 
\begin{itemize}
    \item $E_\cap^+ = \{(u,v) \in E_\cap | \ (\dist(s_1, u) < \dist(s_1, v)) \land (\dist(s_1, u) < \dist(s_1, v)) \}$

    \item $E_\cap^- = \{(u,v) \in E_\cap | \ (\dist(s_1, u) < \dist(s_1, v)) \land (\dist(s_1, u) > \dist(s_1, v)) \}$
\end{itemize} 

\paragraph{Handling Agreeing Paths}
We first address the case of agreeing paths; disagreeing paths can be analyzed analogously.
Consider the subgraph $\G^+ = (V, E_\cap^+)$ of the input graph $G$ and let $\B^+$ represent the set of connected components of this undirected graph $\G^+$.
Using this, we present the Algorithm~\ref{algo:min-2-DSP-Undirected-agreeing}, which finds the agreeing shortest paths with minimum number of intersections. For brevity's sake, we skip the proof of correctness of the algorithm, as it is almost identical to the proof of the algorithm for DAGs (Algorithm~\ref{algo:min-2-DSP-DAG}).

\begin{algorithm}[!ht]
\setstretch{1.25}
\tcc{Assuming minimum number of intersections is at least two}

Compute $\G^+=(V,E_\cap^+)$\;
Let $\B^+$ be a partitioning of $V$ induced by the connected components of $\G^+$\;
\ForEach{$B\in \B$}{
    Set $\tau$ as ordering of vertices in $B$ in increasing distance from $s_1$\;
    Initialize $H_B$ to $\G[B]$, and direct the edges of $H_B$ according to the ordering $\tau$\; 
    Add 6 additional vertices $s_B, s_{B,1}, s_{B,2},t_B,t_{B,1}, t_{B,2}$ to $H_B$\;
    \ForEach{$i \in \{1,2\}$}{
        \ForEach{$(x,y) \in E(s_i, t_i),x \notin B, y \in B$ and $\dist(s_i, x) + \dist(x, y) = \dist(s_i, y)$}{
            Add the directed edge $(s_{B,i}, y)$ to $H_B$\; 
        }

        \ForEach{$(x,y) \in E(s_i, t_i),x \in B, y \notin B$ and $\dist(s_i, x) + \dist(x, y) = \dist(s_i, y)$}{
            Add the directed edge $(x, t_{B,i})$ to $H_B$\; 
        }

        Add the directed edges $(s_B, s_{B,i}), (t_{B,i}, t_B)$ to $H_B$\;
    }

    Compute $T_{out}=\domtree(s_B,H_B)$\;
    Compute $T_{in}=\domtree(t_B,(H_B)^{rev})$\;
    \lForEach{$v \in B$}{
        Set $\alpha(v) = \depth(v,T_{in}) + \depth(v,T_{out})-1$
        }
}

Return $\min_{v\in V\setminus\{s_1,t_1,s_2,t_2\}} \alpha(v)$
\caption{\textsc{Min-2-DSP-Undirected-Agreeing}($G, s_1, t_1, s_2, t_2$)}
\label{algo:min-2-DSP-Undirected-agreeing}
\end{algorithm}

\paragraph{Time Complexity} We solve single-source shortest path problem in $O(m + n)$ time for undirected graphs~\cite{Thorup97}. Rest of the algorithm takes $O(m + n)$ time. Hence, the overall procedure takes $O(m + n)$ time.

This results in the following lemma.

\begin{lemma}   
\label{lemma:twoDSPUndirected-1}
    Given that minimum number of intersections between any pair of shortest paths is at least two, there exists an algorithm which finds agreeing paths with minimum number of intersections and reports the explicit paths in $O(m + n)$ time. 
\end{lemma}

\paragraph{Handling Disagreeing Paths} Similarly, the minimum number of intersections among disagreeing paths can be found using the subgraph $\G^- = (V, E_\cap^-)$ and the partition $\B^-$ induced by the set of connected components of this undirected graph $\G^-$. An alternate easier way of thinking about disagreeing paths is switching the terminals $s_2, t_2$ with each other, i.e., $(s_2', t_2') = (t_2, s_2)$. Now, all disagreeing $(s_1, t_1), (s_2, t_2)$ shortest paths become agreeing paths instead, and hence, can be handled easily. This results in the following lemma:

\begin{lemma}
\label{lemma:twoDSPUndirected-2}
    Given that minimum number of intersections between any pair of shortest paths is at least two, there exists an algorithm which finds disagreeing paths with minimum number of intersections and reports the explicit paths in $O(m + n)$ time. 
\end{lemma}

Since any pair of shortest paths with at least two intersections are either agreeing or disagreeing, from \Cref{lemma:twoDSPUndirected-1} and \Cref{lemma:twoDSPUndirected-2} we get the following result:

\begin{lemma}
\label{lemma:twoDSPUndirected-3}
    Given that minimum number of intersections between any pair of shortest paths is at least two, there exists an algorithm which finds pair of paths with minimum number of intersections and reports the explicit paths in $O(m + n)$ time.
\end{lemma}

\subsection{Handling Pairs with Singular Intersection}

We now study the scenario where the intersection between $(s_1, t_1)$-shortest path and $(s_2, t_2)$-shortest path is a single node. For any vertex $v$, let $F_v$ be enumerating polynomial for pairs $(P_1, P_2) \in \Pi(s_1, t_1) \times \Pi(s_2, t_2)$ for which $v$ is unique intersection point.


\cite{AkmalVW24} implicitly showed how to evaluate $F_v$, for each $v$, in just $O(m+n\log n)$ time. For completeness, we present the details below.

\smallskip

For any $x,y,w$, let $R_w(x,y)$ be enumerating polynomial of all $(x,w)$ and $(y,w)$ shortest paths that are internally vertex-disjoint.
We can express $R_w(x,y)$ as follows:

\begin{equation*}
\label{Eq:Rwxy}
R_w(x,y) ~=~ F(x,w) F(y,w) ~~- \sum_{a \in N(w) \cap V(x, w) \cap V(y, w)} z_{aw}^2 F(x, a) F(y, a).
\end{equation*}

Now consider a vertex $v$.
Let $\A_v$ be collection of all pair $(P_1, P_2) \in \Pi(s_1, t_1) \times \Pi(s_2, t_2)$ passing through $v$ that satisfies the condition that edges appearing before $v$ in $P_1,P_2$ are distinct, and likewise, the edges appearing after $v$ in $P_1,P_2$ are distinct.\footnote{We say that an edge $e$ appears before $v$ in $P_i$ if  $e$ lie in the subpath $P_i[s_i,v]$. Similarly, we define to notion of an edge $e$ appearing after $v$ in $P_i$.}
Note that $\A_v$ will consists of dis-agreeing pairs as well as pairs with single intersection.

\medskip

For any $v\in V(s_1,t_1)\cap V(s_2,t_2)$,
the enumerating polynomial $A_v$ for pairs in $\A_v$ is given by

\begin{equation}
\label{Eq:Av}
A_v ~=~ R_v(s_1,s_2) R_v(t_1,t_2).
\end{equation}

as the subpaths $P_1[s_1,v]$, $P_2[s_2,v]$ are internally vertex-disjoint, as well as the subpaths $P_2[v,t_1]$, $P_2[v,t_2]$ are internally vertex-disjoint.

\paragraph{Prefix-Suffix overlap}
Let $\X_v\subseteq \A_v$ be set of those pairs $P_1,P_2$ that satisfy 
that the edges incident to $v$ in the subpaths $P_1[s_1,v]$ and $P_2[v,t_2]$ are identical.
Then, for any $v\in V(s_1,t_1)\cap V(s_2,t_2)$,
the enumerating polynomial $X_v$ for pairs in $\X_v$ is given by

\begin{equation}
\label{Eq:Xv}
X_v ~=~ \sum_{(a,v)\in E(s_1,v)\cap E(v,t_2)}
z_{av}^2 \ F(s_1,a)F(v,t_1) \ F(s_2,v)F(a,t_2).
\end{equation}

\paragraph{Prefix-Suffix Disjointness with Suffix-Prefix overlap}
Let $\Y_v\subseteq \A_v$ be set of those pairs $P_1,P_2$ that satisfy 
that (i) the subpaths $P_1[s_1,v]$ and $P_2[v,t_2]$ are internally vertex-disjoint, and
(ii) the edges incident to $v$ in the subpaths $P_1[v,t_1]$ and $P_2[s_2,v]$ are identical.
Then, for any $v\in V(s_1,t_1)\cap V(s_2,t_2)$,
the enumerating polynomial $Y_v$ for pairs in $\Y_v$ is given by

\begin{equation}
\label{Eq:Yv}
Y_v ~=~ \sum_{(v,a)\in E(v,t_1)\cap E(s_2,v)}
R_v(s_1,t_2) \ 
z_{va}^2 \ F(a,t_1) \ F(s_2,a).
\end{equation}

\paragraph{Paths with Single Intersection}
Observe that the non-twin crossing pairs in $\A_v\setminus (\X_v\cup \Y_v)$ is precisely the collection of those pairs for which $v$ is the unique intersection point.
Thus, the enumerating polynomial $F_v$ for pairs $(P_1, P_2) \in \Pi(s_1, t_1) \times \Pi(s_2, t_2)$ with $v$ as unique intersection is given by
$$F_v~=~A_v - (X_v+Y_v).$$

\bigskip

In order to check if $F_v$ is non-zero it suffices to evaluate this polynomial  on a random assignment of the variables $\{z_e\}_{e \in E}$ drawn from a sufficiently large finite field $\mathbb{F}$.

\medskip

Observe that the polynomials $\bar R_w(x,y)$ can be evaluated, for all distinct $x,y\in \{s_1,t_1,s_2,t_2\}$, in $O(m + n)$ total time. 
This is possible because we can evaluate $\barF(x, y)$, for all $(x,y)$ in $\{s_1, s_2\}\times V$ and $V\times \{t_1, t_2\}$, in $O(m+n)$ total time using Lemma~\ref{lem:Fxy} along with the $O(m+n)$ time single-source shortest-path algorithm~\cite{Thorup97} for the terminals. 

\medskip

Given evaluations $\bar R_w(x,y)$ and $\barF(x,y)$ for appropriate $x,y,w$, Eq. \ref{Eq:Av}, \ref{Eq:Xv}, \ref{Eq:Yv}, imply that one can evaluate $\barF_v$, for each $v\in V$, in $O(m+n)$ total time.
By~Schwartz-Zippel lemma, with high probability, Min-2-DSP for input pair $(s_1,t_1)$ and $(s_2,t_2)$ is 1 if and only if $\barF_v\neq 0$, for some $v\in V$. Furthermore, this procedure only takes $O(m+n)$ time overall.

\paragraph{Reporting the paths from the singular point of intersection} After finding the vertex $v$ (if it exists) such that $\barF_v \neq 0$, we describe how to report $(P_1, P_2) \in \Pi(s_1, t_1) \times \Pi(s_2, t_2)$ with $V(P_1) \cap V(P_2) = \{v\}$. It can be easily observed that this is equivalent to reporting the $(s_1, v), (v, t_1), (s_2, v), (v, t_2)$ shortest paths intersecting only at $v$. 

This can be solved by finding shortest paths from $v$ to $s_1, t_1, s_2, t_2$ that are all vertex-disjoint.
Construct a directed graph $G_v$ from $G$ by including for each $(a,b) \in E(G)$ a corresponding directed edge if
$$
\dist(v,a) + wt(a,b) = \dist(v,b).
$$
Observe that there is a one-to-one correspondence between $v$ to $x$ paths in $G_v$ and  $v$ to $x$ shortest paths in $G$, for each $x\in V$.
Finally, add a sink $t$ to $G_v$ along with edges $(s_1, t)$, $(s_2, t)$, $(t_1, t)$, $(t_2, t)$. Now, finding four vertex-disjoint paths from $v$ to $t$ (which can be done in $O(m + n)$ time using the Ford-Fulkerson algorithm) yields $(s_1,t_1)$ and $(s_2,t_2)$ shortest paths in $G$ that intersect only at $v$.

This results in the following lemma.

\begin{lemma}
\label{lemma:twoDSPUndirected-4}
    There exists an algorithm which verifies in $O(m+n)$ time if the minimum number of intersections over all pairs of $(s_1,t_1)$ and $(s_2,t_2)$ shortest paths is exactly 1, and reports the explicit paths which intersect exactly once (if they exist) in $O(m + n)$ time.
\end{lemma}

On combining \Cref{lemma:twoDSPUndirected-3} and \Cref{lemma:twoDSPUndirected-4}, we get the following result.

\mintwoDSPundirected*

\bibliographystyle{alpha}
\bibliography{ref}

\begin{thebibliography}{AWW24}

\bibitem[Akh20]{Akhmedov20}
Maxim Akhmedov.
\newblock Faster 2-disjoint-shortest-paths algorithm.
\newblock In Henning Fernau, editor, {\em Computer Science - Theory and Applications - 15th International Computer Science Symposium in Russia, {CSR} 2020, Yekaterinburg, Russia, June 29 - July 3, 2020, Proceedings}, volume 12159 of {\em Lecture Notes in Computer Science}, pages 103--116. Springer, 2020.

\bibitem[AWW24]{AkmalVW24}
Shyan Akmal, Virginia~Vassilevska Williams, and Nicole Wein.
\newblock {Detecting Disjoint Shortest Paths in Linear Time and More}.
\newblock In Karl Bringmann, Martin Grohe, Gabriele Puppis, and Ola Svensson, editors, {\em 51st International Colloquium on Automata, Languages, and Programming (ICALP 2024)}, volume 297 of {\em Leibniz International Proceedings in Informatics (LIPIcs)}, pages 9:1--9:17, Dagstuhl, Germany, 2024. Schloss Dagstuhl -- Leibniz-Zentrum f{\"u}r Informatik.

\bibitem[BFG25]{BentertFG25}
Matthias Bentert, Fedor~V. Fomin, and Petr~A. Golovach.
\newblock {Tight Approximation and Kernelization Bounds for Vertex-Disjoint Shortest Paths}.
\newblock In Olaf Beyersdorff, Micha{\l} Pilipczuk, Elaine Pimentel, and Nguyen~Kim Thang, editors, {\em 42nd International Symposium on Theoretical Aspects of Computer Science (STACS 2025)}, volume 327 of {\em Leibniz International Proceedings in Informatics (LIPIcs)}, pages 17:1--17:17, Dagstuhl, Germany, 2025. Schloss Dagstuhl -- Leibniz-Zentrum f{\"u}r Informatik.

\bibitem[BH14]{BjorkLundH14}
Andreas Bj{\"o}rklund and Thore Husfeldt.
\newblock Shortest two disjoint paths in polynomial time.
\newblock In Javier Esparza, Pierre Fraigniaud, Thore Husfeldt, and Elias Koutsoupias, editors, {\em Automata, Languages, and Programming}, pages 211--222, Berlin, Heidelberg, 2014. Springer Berlin Heidelberg.

\bibitem[BK17]{BercziK17}
Kristof Berczi and Yusuke Kobayashi.
\newblock {The Directed Disjoint Shortest Paths Problem}.
\newblock In Kirk Pruhs and Christian Sohler, editors, {\em 25th Annual European Symposium on Algorithms (ESA 2017)}, volume~87 of {\em Leibniz International Proceedings in Informatics (LIPIcs)}, pages 13:1--13:13, Dagstuhl, Germany, 2017. Schloss Dagstuhl--Leibniz-Zentrum fuer Informatik.

\bibitem[BNRZ21]{BentertNRZ21}
Matthias Bentert, Andr{\'{e}} Nichterlein, Malte Renken, and Philipp Zschoche.
\newblock Using a geometric lens to find k disjoint shortest paths.
\newblock In Nikhil Bansal, Emanuela Merelli, and James Worrell, editors, {\em 48th International Colloquium on Automata, Languages, and Programming, {ICALP} 2021, July 12-16, 2021, Glasgow, Scotland (Virtual Conference)}, volume 198 of {\em LIPIcs}, pages 26:1--26:14. Schloss Dagstuhl - Leibniz-Zentrum f{\"{u}}r Informatik, 2021.

\bibitem[CTW24]{ChitnisTW24}
Rajesh Chitnis, Samuel Thomas, and Anthony Wirth.
\newblock Lower bounds for approximate (\& exact) k-disjoint-shortest-paths, 2024.

\bibitem[ET98]{Tzoreff98}
Tali Eilam-Tzoreff.
\newblock The disjoint shortest paths problem.
\newblock {\em Discrete Applied Mathematics}, 85(2):113--138, June 1998.

\bibitem[FHW80]{FortuneHopcroftWyllie80}
Steven Fortune, John~E. Hopcroft, and James Wyllie.
\newblock The directed subgraph homeomorphism problem.
\newblock {\em Theoretical Computer Science}, 10(2):111--121, 1980.

\bibitem[FKU24]{FunayamaKU24}
Ryo Funayama, Yasuaki Kobayashi, and Takeaki Uno.
\newblock Parameterized complexity of finding dissimilar shortest paths, 2024.

\bibitem[GT05]{GeorgiadisT05}
Loukas Georgiadis and Robert~Endre Tarjan.
\newblock Dominator tree verification and vertex-disjoint paths.
\newblock In {\em Proceedings of the Sixteenth Annual {ACM-SIAM} Symposium on Discrete Algorithms, {SODA} 2005, Vancouver, British Columbia, Canada, January 23-25, 2005}, pages 433--442. {SIAM}, 2005.

\bibitem[Joh77]{Johnson77}
Donald~B. Johnson.
\newblock Efficient algorithms for shortest paths in sparse networks.
\newblock 24(1):1–13, January 1977.

\bibitem[Loc21]{Lochet21}
William Lochet.
\newblock A polynomial time algorithm for the \emph{k}-disjoint shortest paths problem.
\newblock In D{\'{a}}niel Marx, editor, {\em Proceedings of the 2021 {ACM-SIAM} Symposium on Discrete Algorithms, {SODA} 2021, Virtual Conference, January 10 - 13, 2021}, pages 169--178. {SIAM}, 2021.

\bibitem[LT79]{LT79}
Thomas Lengauer and Robert~Endre Tarjan.
\newblock A fast algorithm for finding dominators in a flowgraph.
\newblock {\em {ACM} Trans. Program. Lang. Syst.}, 1(1):121--141, 1979.

\bibitem[MMPS]{MariMPS24}
Mathieu Mari, Anish Mukherjee, Michał Pilipczuk, and Piotr Sankowski.
\newblock {\em Shortest Disjoint Paths on a Grid}, pages 346--365.

\bibitem[ORS93]{Ogier93}
R.G. Ogier, V.~Rutenburg, and N.~Shacham.
\newblock Distributed algorithms for computing shortest pairs of disjoint paths.
\newblock {\em IEEE Transactions on Information Theory}, 39(2):443--455, 1993.

\bibitem[PSW25]{PilipczukSW25}
Michał Pilipczuk, Giannos Stamoulis, and Michał Włodarczyk.
\newblock Planar disjoint shortest paths is fixed-parameter tractable, 2025.

\bibitem[SM05]{SrinivasM2005}
Anand Srinivas and Eytan Modiano.
\newblock Finding minimum energy disjoint paths in wireless ad-hoc networks.
\newblock {\em Wireless Networks}, 11(4):401--417, Jul 2005.

\bibitem[Tho80]{Thompson80}
Clark~David Thompson.
\newblock {\em A complexity theory for VLSI}.
\newblock PhD thesis, USA, 1980.
\newblock AAI8100621.

\bibitem[Tho97]{Thorup97}
M.~Thorup.
\newblock Undirected single source shortest paths in linear time.
\newblock In {\em Proceedings 38th Annual Symposium on Foundations of Computer Science}, pages 12--21, 1997.

\end{thebibliography}

\appendix
\section*{Appendix}

\section{Omitted Proofs}
\label{sec:deferred-proofs}

\lemorderintersect*
\begin{proof}
Suppose for the sake of contradiction that there are vertices $w_1, w_2 \in S$ such that $w_1$ appears before $w_2$ in $P_1$, but after $w_2$ in $P_2$. Since each sub-path of $P_1$ is also a shortest path, we see that $\dist(w_1, v) > \dist(w_2, v)$, but the reverse inequality is implied by the position of $w_1$ and $w_2$ in $P_2$. This leads to a contradiction. 
\end{proof}

\smallskip
\lemFxy*
\begin{proof}
We first compute the all-pairs distance matrix using Johnson's algorithm~\cite{Johnson77} in $O(mn \log n)$ time. Using this matrix, we evaluate $F(x,y)$ for all $x,y \in V$ as follows: 

\begin{equation}
F(x,x) = 1, \quad
F(x,y) = \sum_{\substack{ (a,y) \in E(x,y)}} F(x,a) \cdot z_{ay}, 
\label{eq:Fxy}
\end{equation}

where the membership $(a,y) \in E(x,y)$ can be verified in constant time using the distance matrix. To compute $F(x, y)$ efficiently for all vertex pairs, we fix a source vertex $x$ and construct the DAG $G_x = (V, \bigcup_{v \in V} E(x, v))$. We then process the vertices of $G_x$ in topological order, ensuring that for each vertex $y$, all required values $F(x, a)$ in Eq.~\ref{eq:Fxy} have already been computed prior to evaluating $F(x, y)$. Since each edge in $E$ is processed once per source $x$, the computation of $F(x, y)$ values for all pairs can be performed in $O(mn)$ time, given the distance matrix of $G$. Therefore, the overall computation takes $O(mn \log n)$ time.
\end{proof}

\begin{lemma}
\label{lemma:concordant-pairs-0}
Let $(P_1,P_2)\in \Pi(s_1,t_1)\times \Pi(s_2,t_2)$ be a pair of paths, and $(v,w)$ be a pair of vertices such that $v$ precedes $w$ on $P_1,P_2$. Then, any internal vertex of $P_1[v,w]$ lying on $P_2$ must be contained in the subpath $P_2[v,w]$.
\end{lemma}

\begin{proof}
Assume, for the sake of contradiction, that there exists an internal vertex $x$ in $P_1[v,w]$ that lies on $P_2$ but is not contained in the subpath $P_2[v,w]$. Without loss of generality, suppose that $x$ appears before $v$ on $P_2$. 

Let $d_1$ and $d_2$ respectively denote the lengths of the subpaths $P_2[x,v]$ and $P_2[v,w]$ (see \Cref{fig:fig-concordant-pair}). Observe that both $P_1[x,w]$ and $P_2[x,w]$ are shortest paths from $x$ to $w$, and hence each has length $d_1 + d_2$. Similarly, $P_1[v,w]$ and $P_2[v,w]$ are shortest paths from $v$ to $w$, so each has length $d_2$. Therefore, the subpath $P_1[v,x]$ would have length $-d_1$, implying that the concatenation $P_1[v,x] \circ P_2[x,v]$ forms a cycle of zero length. This contradicts the assumption that $G$ contains no zero-length cycles.
\end{proof}

\begin{figure}[!ht]
\centering
\includegraphics[height=40mm]{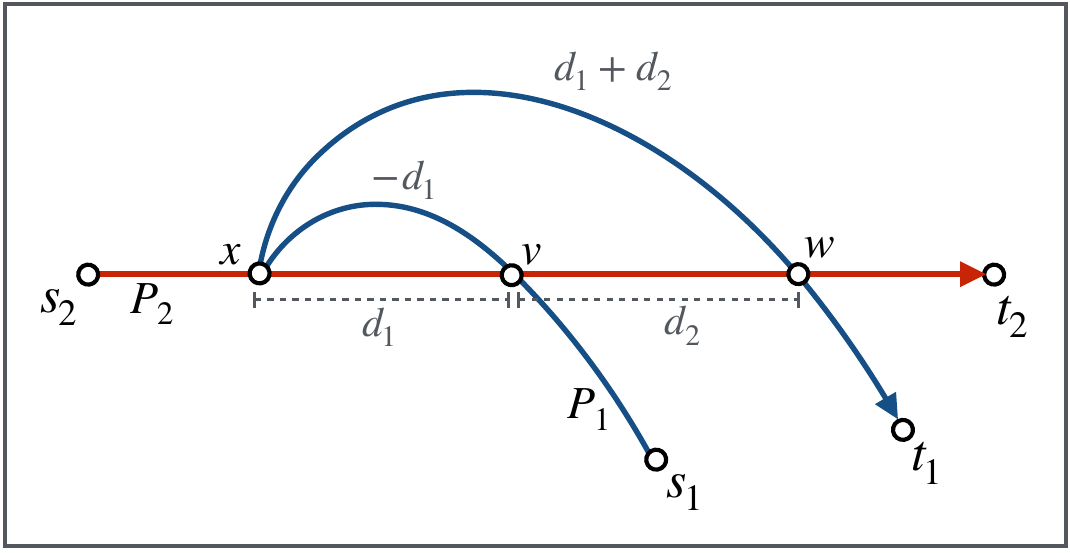}
\caption{Illustration of counter-example.}
\label{fig:fig-concordant-pair}
\end{figure}

\lemconcordantpairs*
\begin{proof}
Let $(v,w),(v',w') \in \mathcal{C}(P_1,P_2)$ be two distinct concordant pairs.  
We first prove claim~1. 
Assume, for contradiction, that the subpaths $P_2[v,w]$ and $P_2[v',w']$ are not vertex-disjoint.  
Without loss of generality, suppose that $v' \in V(P_2[v,w])$ and $v' \neq v$.  
Along $P_1$, there are three possible orderings:
$$
v \prec_{P_1} v' \preceq_{P_1} w 
\quad\text{or}\quad
v' \prec_{P_1} v \prec_{P_1} w
\quad\text{or}\quad
v \prec_{P_1} w \prec_{P_1} v'.
$$
The first ordering cannot be possible, since if $v$ precedes $v'$ on both $P_1$ and $P_2$, then $(v', w')$ cannot be a concordant pair: by definition, no vertex common to $P_1$ and $P_2$ can simultaneously precede $v'$ on both paths. The second and third orderings are ruled out as they contradict \Cref{lemma:concordant-pairs-0}. See~\Cref{fig:vertex-ordering}. Therefore, the subpaths $P_2[v, w]$ and $P_2[v', w']$ (and by symmetry, $P_1[v, w]$ and $P_1[v', w']$) are necessarily vertex-disjoint. This proves the first claim.

\begin{figure}[!ht]
\centering
\includegraphics[height=26.5mm]{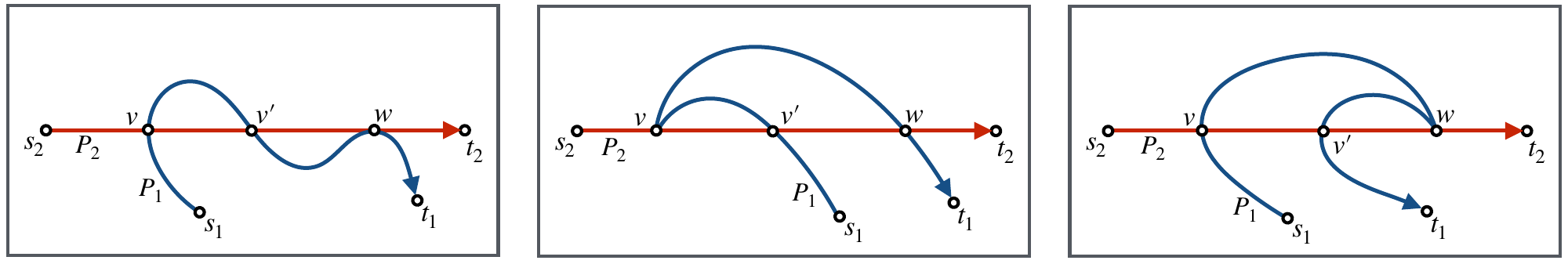}
\caption{Illustration of possible ordering of vertices $v',v,w$ along path $P_1$.}
\label{fig:vertex-ordering}
\end{figure}

We next prove the second claim.
By the definition of concordant pairs, the prefixes $P_1[s_1,v]$ and $P_2[s_2,v]$ as well as the suffixes $P_1[w,t_1]$ and $P_2[w,t_2]$ are internally vertex-disjoint.  
Together with claim~1, this implies that $(v',w')$ must lie entirely within one of the following pairs of segments:
\begin{enumerate}
\item $P_1[s_1,v]$ and $P_2[w,t_2]$, or
\item $P_2[w,t_1]$ and $P_2[s_2,v]$.
\end{enumerate}
Consequently, if $(v',w')$ appears before $(v,w)$ along $P_1$, it must appear after $(v,w)$ along $P_2$, and vice-versa. This establishes the second claim.
\end{proof}

\smallskip
\obspathCritical*
\begin{proof}
Define a subgraph $G^{out}_v$ of $G$ by including only those edges $(a,b) \in E(G)$ that satisfy
$$
\dist(v,a) + wt(a,b) = \dist(v,b).
$$

That is, $G^{out}_v$ contains exactly those edges that participate in some shortest path from $v$ to $b$, for $b\in V$. It follows that $G^{out}_v$ is acyclic, and for any $w\in V$, there is a bijection between $(v,w)$-shortest-paths in $G$ and $(v,w)$-paths in $G^{out}_v$. 

This implies that the cut vertices for pair $(v,w)$ in $G^{out}_v$ correspond exactly to the $(v,w)$-distance critical vertices in $G$.

By applying the vertex-capacitated version of the Ford-Fulkerson algorithm to $G^{out}_v$, we obtain two paths from $v$ to $w$ in $G^{out}_v$, intersecting at only $(v,w)$-cut vertices. These paths correspond to two $(v,w)$-shortest paths in $G$ that intersect only at $(v,w)$-distance critical vertices.
\end{proof}

\smallskip
\lemHv*
\begin{proof}
Consider a pair of shortest paths $( P_1, P_2 ) \in \Pi(x_1,y_1) \times \Pi(x_2,y_2) $ contributing to $H_v(x_1, y_1, x_2, y_2)$. For $i=1,2$, decompose $P_i = Q_i \cdot R_i$, where $Q_i = P_i[x_i, v]$ and $R_i = P_i[v, y_i]$ are shortest paths.

Then,
$$f(P_1)f(P_2) = \prod_{i=1}^2 \big(f(Q_i) \cdot f(R_i)\big).$$

Thus, $H_v(x_1, y_1, x_2, y_2)$ enumerates all such tuples $(Q_1, R_1, Q_2, R_2)$ of shortest paths satisfying the following.
\begin{itemize}
\item For each $i = 1, 2$, $Q_i$ is an $(x_i, v)$-shortest path and $R_i$ is a $(v, y_i)$-shortest path.
\item The paths $Q_1$ and $R_2$ are internally vertex-disjoint.
\end{itemize}

The contribution of all such internally disjoint path pairs $(Q_1, R_2)$ is captured by the polynomial $F_{\disj}(x_1, v, v, y_2)$. Meanwhile, the contributions of the unconstrained paths $Q_2$ and $R_1$, which are only required to be shortest paths between their endpoints, are given by $F(x_2, v)$ and $F(v, y_1)$, respectively. Multiplying these three components yields the desired expression.
\end{proof}

\smallskip
\lemHav*
\begin{proof}
Consider a pair of shortest paths $( P_1, P_2 )$ contributing to $H_{av}(x_1, y_1, x_2, y_2)$. For $i = 1, 2$, decompose each $P_i = Q_i \cdot e \cdot R_i$, where $Q_i = P_i[x_i, a]$ and $R_i = P_i[v, y_i]$ are shortest paths.
Then,

$$
f(P_1) f(P_2) = \prod_{i=1}^2 \big(f(Q_i) \cdot z_{av} \cdot f(R_i)\big) = z_{av}^2 \cdot \prod_{i=1}^2 f(Q_i) f(R_i).
$$

Thus, $H_{av}(x_1, y_1, x_2, y_2)/z_{av}^2$ enumerates all such tuples $(Q_1, R_1, Q_2, R_2)$ of shortest paths satisfying the following.
\begin{itemize}
\item For each $i = 1, 2$, $Q_i$ is an $(x_i, a)$-shortest path and $R_i$ is a $(v, y_i)$-shortest path.
\item The paths $Q_1$ and $R_2$ are internally vertex-disjoint.
\end{itemize}

The contribution of all such internally disjoint path pairs $(Q_1, R_2)$ is captured by the polynomial $F_{\disj}(x_1, a, v, y_2)$. Meanwhile, the contributions of the unconstrained paths $Q_2$ and $R_1$, which are only required to be shortest paths between their endpoints, are given by $F(x_2, a)$ and $F(v, y_1)$, respectively. Multiplying these components together, along with the $z_{av}^2$ term, yields the desired expression.
\end{proof}

\smallskip
\lemDv*
\begin{proof}
Let $\A$ be a collection of all pairs of paths $(P_1,P_2) \in \Pi(x_1, y_1) \times \Pi(x_2, y_2)$ such that
\begin{enumerate}
\item $v\in V(P_1)\cap V(P_2)$,  
\item prefix $P_1[x_1,v]$ and suffix $P_2[v,y_2]$ are internally vertex-disjoint, and
\item  if $a_1$ and $a_2$ are the vertices preceding $v$ in $P_1$ and $P_2$ respectively (assuming these vertices exist, otherwise this condition holds vacuously), then $a_1 \neq a_2$. 
\end{enumerate}

The definitions of the polynomials $H_v$ and $H_{av}$ imply that the enumerating polynomial for the pairs in $\A$ is given by
\[
H_v(x_1, y_1, x_2, y_2)
-
\sum_{a \in \IN(v)} H_{av}(x_1, y_1, x_2, y_2).
\]

Let $\B$ be the set of all pairs of paths $(P_1, P_2)$ contributing to $D_v(x_1, y_1, x_2, y_2)$. In order to prove the claim it suffices to prove that the enumerating polynomials of $\A$ and $\B$ are identical (mod 2).
First, observe that $\B \subseteq \A$. Indeed, if $(P_1, P_2) \in \B$, then $P_1[x_1, v]$ is internally vertex-disjoint from both $P_2[x_2, v]$ and $P_2[v, y_2]$, which in particular ensures that the in-edges of $v$ in $P_1$ and $P_2$ (if they exist) must be distinct.

Now consider any pair $(P_1, P_2) \in \A \setminus \B$. The definitions of $\A$ and $\B$ imply that the following two conditions must hold: (i) the sub-paths $P_1[x_1,v]$ and $P_2[x_2,v]$ have a common internal vertex, (ii) the edges preceding $v$ in these two sub-paths respectively are distinct. 
Let $r$ be the {\bf immediate} predecessor of $v$ common to both the paths. 
It follows that $(r,v)$ is a twin-crossing pair for $(P_1, P_2)$. Let $(Q_1, Q_2) = \phi_{r,v}(P_1,P_2)$.~\Cref{cl:crossing} shows that $(Q_1, Q_2)$ has $(r,v)$ as twin-crossing pair and hence lie in $\A \setminus \B$ (see \Cref{figure:swap-2-DSP}).
Further, it follows that $r$ must be the immediate predecessor of $v$ on $Q_1,Q_2$ and $(P_1, P_2) = \phi_{r,v}(Q_1, Q_2)$.


Thus, the set $\A \setminus \B$ can be partitioned into pairs 
\[
\{(P_1, P_2), (Q_1, Q_2)\},
\]
where $(Q_1, Q_2) = \phi_{r,v}(P_1, P_2)$ and $r$ is the immediate predecessor of $v$ on $P_1,P_2$ (as well as, on $Q_1,Q_2$).  
Since $f(P_1) f(P_2) = f(Q_1) f(Q_2)$, we see that each such pair contributes two identical monomials to the enumerating polynomial of $\A \setminus \B$. Thus, the total contribution of $\A \setminus \B$ is a multiple of two, and hence zero in a field of characteristic two. This establishes the claim. 
\end{proof}

\begin{lemma}
We can compute in $O(mn)$ time an ordered sequence $\sigma$ of the vertices in $V$ that respects the partial order $\prec$, that is, for all $u, v \in V$, if $u \prec v$ then $u$ appears before $v$ in $\sigma$.
\label{lem:partialorder}
\end{lemma}

\begin{proof}
Construct a subgraph $H$ of $G$ as follows: for each edge $(a, b) \in E(G)$, include the edge $(a, b)$ in $H$ if
$$
\dist(s_1,a)+wt(a,b)=\dist(s_1,b)
$$

The computation of graph $H$ takes $O(mn)$ time.
For each vertex $x \in V$, there is a one-to-one correspondence between $s_1$-to-$x$ paths in $H$ and $s_1$-to-$x$ shortest paths in $G$. Furthermore, since $G$ contains no cycles of zero weight, $H$ must be acyclic.
A topological ordering $\sigma$ of $H$ is a total order that extends the partial order $\preceq$.
For any $x, y \in V$ with $x \preceq y$, we have that $x$ lies on an $s_1$-to-$y$ shortest path; thus, $x$ will precede $y$ in the topological ordering. Thus, the ordering of $V$ respecting $\preceq$ can be found in $O(m+n)$ time.
\end{proof}

\lemlstar*
\begin{proof}
By establishing a total order $\sigma$ extending $\prec$ (see \Cref{lem:partialorder}), we can sort $\cL$ into $\cL^*$ in $O(m+n)$ time.
The recurrence relations in the dynamic programming, specifically those for $\bar{F}_d$, $\bar{D}_v$, $\bar{H}_v$, and $\bar{H}_{av}$ require, for any tuple $(s_1, a, v, t_2) \in \cL$, subproblems of the form $(s_1, u, u, t_2)$ for nodes $u \in V(s_1, v)$, or $(s_1, b, u, t_2)$ for edges $(b, u) \in E(s_1, v)$. By definition, all such tuples are elements of $\cL$ and satisfy $u \preceq v$. Consequently, every subproblem required by the recurrence for $(s_1, a, v, t_2)$ appears earlier in $\cL^*$, ensuring that all needed values are available when computing each DP entry.
Therefore, $\cL^*$ can be constructed in $O(mn)$ time.
\end{proof}

\subsection{Success Probability of Algorithms}
\label{sec:success-prob}
In Algorithm~\ref{algo:2-DSP-directed} as well as Algorithm~\ref{algo:min-2-DSP-directed}, we check whether a certain enumerating polynomial is non-zero or not. Note that each monomial term in the final polynomials of both Algorithm~\ref{algo:min-2-DSP-directed} and Algorithm~\ref{algo:2-DSP-directed} is a product of $(s_1, t_1)$ and $(s_2, t_2)$ shortest path terms and hence, has degree at most $2n$. Hence, the degree of the polynomials to be checked is at most $2n$. So, if we take our underlying field of order $2^q$, where $q = 4 \log_2 n$, then, by Schwartz-Zippel lemma, we get that the probability with which our algorithms are correct is at least $1 - \frac{2n}{2^{4 \log_2 n}} \geq 1 - \frac{1}{n^2}$. Also, if we work in Word-RAM model with words of size $O(\log n)$, our basic field operations run in $O(1)$ time.





\end{document}